\documentclass[reqno,11pt]{amsart}
\usepackage[T1]{fontenc} 
\usepackage[utf8]{inputenc} 
\usepackage{mathtools}
\usepackage{amssymb}
\usepackage{amsthm}
\usepackage{bbm}
\usepackage{mathrsfs}
\usepackage[english]{babel} 
\usepackage{enumitem}
\usepackage{dsfont}
\usepackage{url}
\usepackage{hyperref}

\hypersetup{linktocpage=true, colorlinks=true, breaklinks=true, linkcolor=blue, menucolor=black, urlcolor=black,citecolor=blue,filecolor=green}

\newcommand{\N}{\mathbb{N}} 

\theoremstyle{plain}
\newtheorem{theorem}{Theorem}[section]
\newtheorem{proposition}[theorem]{Proposition}
\newtheorem{corollary}[theorem]{Corollary}
\newtheorem{lemma}[theorem]{Lemma}

\theoremstyle{definition}

\newtheorem{remark}[theorem]{Remark}
\newtheorem{example}[theorem]{Example}

\def\eps{\varepsilon}

\def\N{\bbN}

\numberwithin{equation}{section}

\newcommand{\nn}{\nonumber}

\DeclareMathSymbol{\leqslant}{\mathalpha}{AMSa}{"36}
\DeclareMathSymbol{\geqslant}{\mathalpha}{AMSa}{"3E}
\DeclareMathSymbol{\doteqdot}{\mathalpha}{AMSa}{"2B}
\DeclareMathSymbol{\circlearrowright}{\mathalpha}{AMSa}{"08}
\DeclareMathSymbol{\subsetneq}{\mathalpha}{AMSb}{"28}
\DeclareMathSymbol{\supsetneq}{\mathalpha}{AMSb}{"29}
\renewcommand{\leq}{\;\leqslant\;}
\renewcommand{\geq}{\;\geqslant\;}

\newcommand{\dd}{{\rm d}}
\newcommand{\e}[1]{\,{\rm e}^{#1}\,}

\DeclareMathOperator*{\supp}{\text{supp}}

\newcommand{\upchi}{\raise 2pt \hbox{$\chi$}}

\newcommand{\caE}{{\mathcal E}}
\newcommand{\caF}{{\mathcal F}}

\newcommand{\caH}{{\mathcal H}}

\newcommand{\caM}{{\mathcal M}}
\newcommand{\caN}{{\mathcal N}}

\def\bbone{{\mathchoice {\rm 1\mskip-4mu l} {\rm 1\mskip-4mu l} {\rm 1\mskip-4.5mu l} {\rm 1\mskip-5mu l}}}

\newcommand{\bbE}{{\mathbb E}}

\newcommand{\bbN}{{\mathbb N}}

\newcommand{\bbP}{{\mathbb P}}

\newcommand{\bbR}{{\mathbb R}}

\newcommand{\de}{{\mathrm{d}}}
\newcommand{\one}{{\mathbbm{1}}}
\newcommand{\bbPh}{\hat{\bbP}}
\newcommand{\bbPw}{\Tilde{\bbP}}

\def\lt{\left}
\def\rt{\right}

\def\la{\langle}
\def\ra{\rangle}

\title{Enhanced Binding for a quantum particle coupled to scalar quantized field}
\author{Volker Betz}
\author{Tobias Schmidt}
\author{Mark Sellke}
\date{\today}

\begin{document}

\maketitle

\begin{abstract}
    Enhanced binding of a quantum particle coupled to a quantized field 
    means that the Hamiltonian of the particle alone 
    does not have a bound state, 
    while the particle-field Hamiltonian does. For the 
    Pauli--Fierz model, this is usually shown via the binding condition, 
    which works less well in the case of a linear coupling to a scalar 
    field. In particular, the case of a single particle linearly coupled to 
    a scalar field has been open so far. 
    Using a method relying on functional 
    integrals and the Gaussian correlation inequality, we obtain enhanced 
    binding for this case. From a statistical mechanics point of view, our 
    result describes a localization phase transition (in the strength of the pair potential) for a Brownian motion 
    subject to an external and an attractive pair potential.   
\end{abstract}

\section{Introduction and Results}
When a quantum particle is coupled to a quantized field, it
distorts the field around itself. When the particle is
accelerated, this distortion needs to move around with
the particle. As a consequence, more force is needed for a given amount of
acceleration. This effect is paraphrased by saying that the
effective mass $m_\text{eff}$ of the particle increases due to the 
coupling with the field. 

A  way to detect the increased effective mass experimentally would be to 
put the particle into an external potential with compact support that is too weak 
to bind the particle alone, and observe whether a steady state exists when the 
particle is coupled to the field. Mathematically, this amounts to a question 
about the existence of eigenfunctions. Consider a quantum particle of mass $m>0$ 
(not coupled to a field) with 
Hamiltonian $H_{\rm p} = -\frac{1}{2m} \Delta + \delta U$. Here, $U<0$ is a bounded potential with compact support and 
$\delta > 0$. In this case it is well known \cite{LaLi77}
that for $\delta$ small enough, 
no $L^2$-eigenvector (bound state) of $H_{\rm p}$ exists, while for sufficiently large $\delta$, 
bound states can be found. The absence of bound states can also be seen directly by Lieb-Thierring inequalities, see \cite{Sim79}. For the circular finite depth
well of radius $a>0$ the threshold where a bound state 
appears can be calculated and turns out to be equal to 
$\delta_{\rm c} = \frac{\pi^2}{8 m a^2}$ in units where $\hbar = 1$. 
As expected, $\delta_{\rm c}$ is decreasing in the mass $m$ of the particle. 
The task is now to show that for some $\delta < \delta_{\rm c}$, a bound 
exists when the particle is coupled to a quantum field. 
When this happens, one speaks of enhanced binding.

Enhanced binding has received considerable interest in the last 20 years. 
Many works treating enhanced binding rely on the 
{\em binding condition} introduced in the seminal work \cite{GrLiLo01}. 
Of those, the vast 
majority treat the Pauli-Fierz model describing an electron interacting 
with its electro-magnetic field. For other models such as the Fr\"ohlich 
polaron or the Nelson model, the use of the binding condition has been less successful.  
Relevant cases remained open where enhanced 
binding has not been proven so far. 
One of them is the one-particle Nelson 
model. In this paper, we prove enhanced binding for this model, using an alternative approach to the binding 
condition that is based on functional 
integrals and the Gaussian correlation inequality. We postpone 
a comparison of our results with the existing literature, and in particular with the methods using the binding condition, to the next section. Here, we  introduce the model and our results. 

The Nelson model describes a particle linearly coupled to a 
scalar quantum field, such as a charged particle in a polar crystal. Its Hamiltonian acts on $\caH = L^2(\bbR^d) \otimes \caF$, where $\caF$ is the bosonic Fock space, and is given by 
\begin{equation}
    \label{hamiltonian}
 H = H_{\rm p} + H_{\rm f} +  \alpha H_{\rm I}.
\end{equation}
As stated earlier,
\[
H_{\rm p} = -\frac12 \Delta + \delta U
\]
is the Hamiltonian of the quantum particle, where the 
potential $U$ is assumed to be Kato-decomposable
\cite{BeLoHi19}. We will 
mostly be concerned with the case of bounded $U$. The 
operators
\[
H_{\rm f} = \int \de k \omega(k) a^\dagger _k a_k,
\]
and
\[
H_{\rm I} = \int_{\bbR^d} \de k \frac{1}{2\sqrt{\omega(k)}} \lt( \hat{\varphi}(k) \e{-ik\cdot x}a^\dagger (k) + \hat{\varphi}(-k) \e{ik\cdot x} a(k) \rt).
\]
are the field energy and the particle-field interaction energy, respectively. The parameters $\omega:\bbR^d \to \bbR$
and $\varphi: \bbR^d \to \bbR$ are chosen such that $\omega$ 
and $\hat \varphi$ are radially symmetric, with 
 $\hat{\varphi}/ \omega \in L^2(\bbR^d)$ and
 $\hat{\varphi}/ \sqrt{\omega} \in L^2(\bbR^d)$. The 
 symbols $a^\dagger(k)$ and $a(k)$ represent the usual operator-valued distributions describing creation and annihilation operators defined on Fock space, satisfying the canonical commutation relations $[a_k,a^\dagger _{k'}] = \delta (k-k')$. Since we will immediately switch to the functional integral representation, we give no further information on these objects, but rather refer to \cite[Chapter 2]{HiLo20} or 
 \cite{HiSaSpSu12}.
 

The path integral formulation of the particle-field model emerges through a Feynman-Kac formula. We again refer to 
\cite{HiLo20} for a derivation and only state the correspondence here. We denote by $\bbE^\nu$ an expectation with respect to the probability measure $\nu$. In case no measure is given, we have $\bbE = \bbE^\bbP$, where $\bbP$ is the path measure of standard Brownian motion (i.e.\ Wiener measure) in $d$ dimensions. Moreover, $\bbE^x$ is an expectation w.r.t. Wiener measure started at $x \in \bbR^d$. With this notation at hand, take $f,g \in L^2(\bbR^d)$ and let $\Omega$ denote the Fock vacuum. With
\begin{equation}
    \label{field_integrated_pair_potential}
    \overline{W}(x,t) = -\frac{1}{2} \int_{\bbR^d} \de k \frac{| \hat{\varphi}(k)|^2}{2 \omega (k)} \cos (k \cdot x) \e{-|t|\omega(k)}
\end{equation}
for $x \in \bbR^d$, $t \in \bbR$ it holds that
\begin{equation}
    \label{fkn_formula}
    \la f \otimes \Omega , \e{-tH} g \otimes \Omega   \ra = \int_{\bbR^d} \de x f(x) \bbE^x \lt[ \e{ - \alpha \int_0 ^t \de r \int_0 ^t \de s \overline{W}(x_s -x_r, s-r) - \delta \int_0 ^t \de s U(x_s)}  g(x_t) \rt].
\end{equation}

The main technical result of our paper will actually be about 
the path measure $\bbPh_{\delta,\alpha,T}$ defined in \eqref{nelson_path_measure}, which corresponds to the right hand side of 
\eqref{fkn_formula}. In Theorem \ref{recurrence_half_time_theorem} below, we  prove that 
the distribution of $x_{T/2}$ under $\bbPh_{\delta,\alpha,T}$ does not converge vaguely to zero as $T \to \infty$. From this, we can rather easily infer the existence of a bound state. 

Theorem \ref{recurrence_half_time_theorem} is a statement about 
localization: for bounded $U$ and small enough $\delta > 0$, the distribution of $x_{T/2}$ under $\bbPh_{\delta,0,T}$ converges vaguely to zero. In 
this case, localization by self-interaction of the paths is achieved.  
This bears some resemblance to results about 
self-interacting random walks, see e.g.\ \cite{BoSc97} or 
\cite{BrSl95}. In these works, no single site 
potential is present; additionally, the pair interaction
is of mean field type, which implies that long-range interactions are as important as short-range ones. Then, large deviations techniques 
become available, which can not be used in 
our case. In fact, localization for $\delta = 0$ does not occur at any $\alpha$ in the type of models 
that we treat, see e.g.\ \cite{BeSp05}.

Let us now introduce our precise assumptions and results. Below, we will 
often speak of ground states, which are bound states whose eigenvalue is
equal to $\inf {\rm spec} H$. All our methods actually show existence of ground states and thus, a fortiori, of bound states.

Since $\overline{W}$ is radial and $\overline W(x,t) = \overline W(x,-t)$, there exists $W:\bbR_{\geq 0} \times \bbR_{\geq 0} \to \bbR$ so that $\overline W(x,t) = W(\|x\|,|t|)$. We will mostly use $W$ below. We make the following assumptions.

    \begin{enumerate}[label=\textbf{(A.\arabic*)},ref=(A.\arabic*)]
\item 
\label{potential_assumption}
$U: \bbR^d \rightarrow \bbR_-$ is radial and quasi-convex (i.e.\ $\|x\| \mapsto U(x)$ is increasing).  Further, it is non-vanishing in a $\epsilon$-neighborhood around $0$. Finally, $\lim_{\|x\| \rightarrow \infty} U(x) \in (U(y),0]$ with $\Vert y \Vert = \epsilon$.
\item 
\label{pair_potential_assumption}
$\overline W$ is jointly continuous and $x \mapsto \overline{W}(x,t)$ is quasi-concave for all $t \geq 0$; this means that 
$W$ is jointly continuous and
    $$x \mapsto W(x,t)$$
    is decreasing for all $t \geq 0$. 
Moreover, we require the time decay to be sufficiently fast such that 
\begin{equation}
\label{eq:time-decay-integrable}
\sup\limits_{T > 0} \Big| \int_0 ^T \de t\int_T ^\infty \de s ~ W \lt(\Vert x_t - x_s \Vert,|t-s| \rt) \Big| \le C_I < \infty \:\:\:\:\:\:\: \bbP-\text{a.s..}
\end{equation}
\end{enumerate}

For our enhanced binding result some additional decay properties of $W$ near $t=0$ are required.
    \begin{enumerate}[label=\textbf{(B)},ref=(B)]
\item There exists $l^\ast>0$, $R>0$ and $\xi >0$  so that for all $t \in [0, l^*]$, 
\label{local_decrease}
    \begin{equation}
    \label{decreasing_property}
        [0,R] \ni x \mapsto W(x,t) + \frac{1}{\xi} x^2
    \end{equation}
    is decreasing.
\end{enumerate}
Regarding Assumption~\ref{potential_assumption}, the condition $U<0$ together with quasi-convexity implies that $\lim_{\|x\| \rightarrow \infty} U(x) = u \in (U(0),0]$ exists. By a shift of energy we can assume without loss of generality that $u=0$.

\begin{example}
    Let us construct some examples that satisfy our assumptions. In terms of $U$, the prototypical example is $U = - \delta \one_{[0,r]}(\Vert \cdot \Vert)$ for some radius  $r>0$ and depth $\delta > 0$. For the field we present two examples:
    \begin{enumerate}
        \item Take $d \in \bbN$ and set  $\omega \equiv 1$. We are left with evaluating 
        \begin{equation}
            \nn
            \frac{1}{2} \e{-|t|} \int_{\bbR^d} \de k | \hat{\varphi}(k)|^2 \cos (k \cdot x) = \frac{1}{2} \e{-|t|} \int_{\bbR^d} \de k \widehat{\big( \varphi *\varphi\big)} (k)  \e{-i k \cdot x} = \frac{1}{2} \e{-|t|} \big( \varphi * \varphi \big) (x),
        \end{equation}
    since 
    $$| \hat{\varphi}(k)|^2  =  \hat{\varphi}(k)^2 = \widehat{ (\varphi * \varphi)  }(k)$$
    for $\varphi$ sufficiently nice. Now, the convolution of two log-concave functions remains a log-concave function \cite{BrLi74}. This implies that $\varphi * \varphi$ is quasi-concave, and so choosing $\varphi$ log-concave and symmetric with sufficiently fast decay at $\infty$ (take for example $\varphi (x) = \exp(-x^2/2)$) defines a model that we can treat. 
    \item The classical Nelson model in three dimensions is characterized by $d=3$ and $\omega(k) = |k |$. In order to get a pair potential fulfilling \ref{pair_potential_assumption}, we need to choose the cutoff function $\varphi$ carefully. Let us take $\hat{\varphi}(k) = \hat{\Gamma}(k) \sqrt{|k|}$ for $\Gamma$ decaying sufficiently fast. We find that
    \begin{equation}
        \label{nelson_pair_interaction}
        \overline{W}(x,t) = \frac{4 \pi}{|t|^3} \lt( \frac{1}{( \: \cdot \:^2  / t^2 + 1)^2}  * \Gamma * \Gamma  \rt) (x).
    \end{equation}
     By a similar argument as before, this function is quasi-concave in $x$ for each fixed $t$ whenever $\Gamma$ is. 
    Moreover, the infrared regularity property
    is fulfilled, which implies the bound (see equation (2.5.2) in \cite{HiLo20})
    $$\int_0 ^T \de t \int_T ^\infty \de s \big| W(\Vert x_t - x_s \Vert, |t-s|) \big| \le \int_{\bbR^3} \de k \frac{\hat{\varphi}(k)^2}{|k|^3} < \infty.$$
   Taking again $\Gamma(x) = \exp(-x^2 /2)$ as an example, it is seen from (\ref{nelson_pair_interaction}) or (\ref{field_integrated_pair_potential}) that \ref{local_decrease} is fulfilled. 
    \end{enumerate}
\end{example}

\begin{theorem}
    \label{main_thm}
    Assume \ref{potential_assumption}, \ref{pair_potential_assumption} and \ref{local_decrease}.
    Fix $\delta >0$ and let
    $$H(\alpha) = - \frac{\Delta}{2} + \delta U +  H_f + \alpha H_I.$$
     Then, there exists $\alpha^* > 0$ such that $H(\alpha)$ has a ground state for all $\alpha \geq \alpha^*$. 
\end{theorem}

\begin{remark}
\hfill
    \begin{enumerate}
        \item One would naively expect that $U(x)= \delta \one_{[0,\epsilon]}(\Vert x \Vert)$ (for any $\epsilon,\delta >0$) is the only external potential that one needs to treat. Indeed, for any $U$ satisfying \ref{potential_assumption}, there exists $\Tilde{U} = \delta \one_{[0,\epsilon]}(\Vert x \Vert)$ so that $U \ge \Tilde{U}$. While the implication
    $$ - \frac{\Delta}{2}+ \Tilde{U} \text{ has a ground state } \implies - \frac{\Delta}{2}+ U \text{ has a ground state }$$
    follows from classical perturbation theory, we do not know how to similarly argue in case the particle is coupled to the field. The difficulty here is that a spectral gap need not exist. 
        \item Our proof gives explicit asymptotic behaviour for $\alpha$ as $\delta \downarrow 0$. It is shown in Lemma \ref{recurrence_lemma} that $\alpha \ge O(\delta^{-2})$ suffices in case $\delta$ is small.
    \end{enumerate}
\end{remark}
Condition (\ref{eq:time-decay-integrable}) is related to the  infrared regularity of the quantum model. There are physically relevant examples where $H_p$ has a ground state which vanishes as soon as $\alpha >0$ in case of infrared singularity (e.g. Theorem 2.63 in \cite{HiLo20} and \cite{HiMa22} for a detailed discussion). Moreover condition \ref{local_decrease} is not very restrictive, and can be relaxed further in exchange for additional assumptions on $H_p$.
   
\begin{proposition}
\label{simple_gs}
    Fix $\delta >0$, assume \ref{potential_assumption}, \ref{pair_potential_assumption}and instead of \ref{local_decrease} assume that $H_p$ has a spectral gap. Let
    $$H(\alpha) = - \frac{\Delta}{2} + \delta U +  H_f + \alpha H_I.$$
     Then, $H(\alpha)$ has a ground state for all $\alpha \geq 0$.
\end{proposition}

Let us end this first section by giving a brief proof summary and by discussing the relations to statistical mechanics. 
Our assumptions guarantee that we can apply the method of  \textit{Gaussian domination} introduced recently by one of the authors in \cite{Se22} (see also \cite{BeSchSe23}). This technique allows us to replace the path measure induced by the Nelson model with a Gaussian one for the purpose of estimating fluctuations. This makes explicit calculations possible. A key observation here is that this simpler measure lets us trade the field coupling strength $\alpha$ for a larger multiplicative constant $\Tilde{\delta}(\alpha) > \delta$ on the external potential (see the proof of Lemma \ref{recurrence_lemma}). We then deduce from this that the particle stays close to the origin with positive probability at time $T/2$. This is the content of the following Theorem. For the definition of $\bbPh_{\delta,\alpha,T}$ we refer to equation (\ref{nelson_path_measure}).

\begin{theorem}
\label{recurrence_half_time_theorem}
    Fix $\delta >0$ and
    assume \ref{potential_assumption}, \ref{pair_potential_assumption}. In addition, assume  \ref{local_decrease} or that $H_p$ has a spectral gap.
    Then, there exists $\alpha^* > 0$, $c>0$ and a symmetric and compact set $K$ such that for all $\alpha \ge \alpha^*$
    $$\liminf\limits_{T \rightarrow \infty} \bbPh_{\delta,\alpha,T} \big( x_{T/2} \in K \big)  > c.$$
\end{theorem}

Knowing the statement of Theorem \ref{recurrence_half_time_theorem}, it is a consequence of spectral theory and the quasi-concavity of $U,W$ that a unique strictly positive ground state for the corresponding model exists; see Theorem \ref{ground_state_options}. While this result is known in the literature in a variety of flavours, we provide a proof for the reader's convenience. 

Knowing the existence (and uniqueness) of a ground state $\Psi$, 
we can derive the limiting distribution of $x_{T/2}$ as $T \to \infty$.
Spectral theory guarantees the convergence 
\[
\lim_{t \to \infty} \e{-t H} g \otimes \Omega
= \langle g \otimes \Omega, \Psi \rangle \Psi
\]
for all $g \in L^2(\bbR^d)$ . 
Smoothing properties of the semigroup 
also allow the choice of $g = \delta_0$, but if one wants to 
avoid technical difficulties, the easiest way is to replace 
$\bbPh_{\delta,\alpha,T}$ by a measure where also the 
starting point of the Brownian motion is smeared out. All 
arguments will still work, indeed see Propositions \ref{zero_domination} and \ref{prop_zero_max}. In any case, we get 
for any bounded $f \in L^2(\bbR^d)$ that

\[
\begin{split}
& \lim_{T \to \infty} \langle \delta_0 \otimes \Omega, \e{-TH/2} (f \otimes  \Omega) \e{-TH/2} (g \otimes \Omega) \rangle \\
= \,\, &  
\lim_{T \to \infty}
\langle \e{-TH/2} (\delta_0  \otimes \Omega), (f \otimes \Omega) \e{-TH/2} (g \otimes \Omega) \rangle 
\\ 
= \,\, & 
\langle \delta_0 \otimes \Omega, \Psi \rangle \langle 
g \otimes \Omega, \Psi \rangle \langle \Psi, (f \otimes \Omega) \Psi \rangle.
\end{split}
\]
By the Feynman-Kac formula, we have 
\[
\bbPh_{\delta, \alpha,T}(f(x_{T/2})) = 
\frac{\langle \delta_0 \otimes \Omega, \e{-TH/2} (f\otimes \Omega) \e{-TH/2} (\bbone_B \otimes \Omega) \rangle}{\langle \delta_0 \otimes \Omega,  \e{-TH} (\bbone_B \otimes \Omega) \rangle}, 
\]
and we see that the distribution of $x_{T/2}$ under 
$\bbPh_{\delta, \alpha,T}$ converges to $\psi^2(x) \dd x$, 
where $\psi^2$ is $\Psi^2$ with the Fock space component integrated out, defined by the equation 
$\langle \psi, f \psi \rangle_{\bbR^3} = 
\langle \Psi, (f \otimes \Omega) \Psi \rangle_{L^2(\bbR^3) \otimes \caF}$.

From similar considerations, one can see that in cases where $H_p$ has no bound state, we have 
 $$
 \limsup\limits_{T \rightarrow \infty} \bbPh_{\delta,0,T} \big( x_{T/2} \in K \big)  = 0.
 $$

Let us also mention that the proofs of Proposition \ref{simple_gs} and Theorem \ref{main_thm} are identical, up to the first Lemma. These Lemmata (Lemma \ref{recurrence_lemma} and Lemma \ref{b_exp_growth}) allow us to (essentially) deduce exponential growth of (\ref{fkn_formula}), which suffices to show that the particle paths perturbed by field and external potential spend at least some fraction of time close to the origin. Having this property, assumptions \ref{potential_assumption} and \ref{pair_potential_assumption} suffice to extend this confinement result to the path at time $T/2$, which yields the existence of a ground state.

The remainder of this article is organized as follows: in Section 
\ref{sec: binding condition}  we review some of the literature on the enhanced binding problem, and place our contribution into this context. In Section \ref{gci_intro} we introduce the necessary notation and the main tools for our proofs, mainly the Gaussian correlation inequality (GCI) and its consequences. In Section \ref{first_confinement} we use either set of assumptions to derive a proof for Theorem \ref{first_recurrence_theorem}, which is a preliminary confinement result. We stress once again that \textit{Lemma \ref{recurrence_lemma} is the only place where \ref{local_decrease} is needed additionally instead of just \ref{pair_potential_assumption}}. We continue in Section \ref{section_proof_recurrence_half_time_theorem} with the proof of Theorem \ref{recurrence_half_time_theorem}, which is then used in the next section to prove Theorem \ref{main_thm} (and Proposition \ref{simple_gs}). In Section \ref{approx_markov_section} we provide a proof for Proposition \ref{approx_markov}.

\section{Other approaches for showing the existence of bound states} \label{sec: binding condition}
\subsection{The binding condition}
The binding condition was first discussed in \cite{GrLiLo01}. It 
reduces the question about the existence of ground states of particle-field-systems 
to a comparison between the infima of the spectrum of different operators.  
In our context it says that checking the inequality  
\[
\inf {\rm spec} H < \inf {\rm spec} H_{\rm f} + \inf {\rm spec} H_p 
\]
is sufficient for proving the existence of a ground state of $H$. 
The original paper \cite{GrLiLo01} only deals with the 
Pauli-Fierz model of quantum electrodynamics, which differs from 
the Nelson model in that the quantum field is a vector field and 
couples to the momentum of the particle, not its position. 
In \cite{HiMa22}, Hiroshima 
and Matte show that the binding condition is sufficient for the
existence of ground states in the Nelson model as well. While the 
original purpose of the binding condition was to prove that ground 
    states {\em persist} when the particle is coupled to a field, it 
was soon noticed that this technique can also be used to prove 
enhanced binding. 

The first paper \cite{HiSp01} on enhanced binding still does not use the 
binding condition, but also does not treat the full Pauli-Fierz model, instead 
using the dipole approximation. 
Then, it is possible to find a unitary 
transformation such that $H_{\text{PF}} \approx -\frac{1}{2 m_\text{eff} } \Delta + 
U$ is essentially a rigorous statement. With correct assumptions on $U$ 
and an explicit formula for $m_\text{eff}$ increasing in the coupling 
strength, the existence of a ground state is then seen directly. 

In the context of enhanced binding, the binding condition was first used in  
\cite{HaVoVu03}, where the result is proved in the limit of small enough coupling 
constant. In our notation, the statement is that for sufficiently small $\alpha$, 
one can find $\delta > 0$ such that for this $\delta$, $H(\alpha)$ 
has a ground state but $H(0)$ does not. This result was improved in various ways
in the years after, until K\"onenberg and Matte 
\cite{KoMa13} used a variant of an argument from \cite{ChVoVu03} and proved  
enhanced binding in the Pauli-Fierz model for all coupling 
constants, also in the semi-relativistic case. This settled the 
problem of enhanced binding in the Pauli-Fierz model. For  
references on the above mentioned intermediate results, we refer to the introduction 
of \cite{KoMa13}. 

For the Nelson model, on the other hand, verifying the binding 
condition is more difficult, as noted by Hiroshima and Sasaki 
\cite{HiSa07}. In this paper, the authors consider the Nelson model with 
$N \geq 2$ particles. A dressing transform is used in order to extract 
effective attractive potentials between the particles that originate in 
their interaction with the field. Enhanced binding in this case is a 
consequence of this attractive interaction: the particles behave like a 
single, heavier particle, which can be trapped by the external potential. 
The case $N=1$ can not be treated by this argument, and the result describes a different effect from the one 
that we started with: instead of a phonon cloud moving with the 
particle and making it heavier, the particles move in synchronization and thereby 
increase their mass. 

Another paper where the Nelson model is treated using the binding condition 
is the recent work \cite{FaOlRo23}. Here the authors use a version of the 
binding condition in the one-particle Nelson model. To make this approach work, the authors need to 
actually perform the strong 
coupling limit $\alpha \rightarrow \infty$ in order to deduce the existence 
of a ground state. No statement is made for finite $\alpha$. 
It is therefore unclear at the moment whether it is possible to 
adapt the purely spectral methods of the binding condition to cover the one-
particle
Nelson model in general. However, an approach based on ideas from 
\cite{LiSe13} combined with recent results on the effective mass seems to  give at least  
partial results. We sketch this approach after recalling the concept of 
effective mass. 

\subsection{The effective mass}

The concept of effective mass is much older than enhanced binding. It goes back 
at least to the work of Landau and Pekar \cite{LP48} in the context of the 
Fr\"ohlich polaron \cite{Fr37,Fr54}. The approach is a bit less direct than 
the one described above. Also, in this Section only, we will use the convention of the functional analytic community and define $H(\alpha)$ with a bare mass of $1/2$, i.e.\ we set
\[
H(\alpha) = -\Delta + \delta U + H_{\rm f} + \alpha H_I.
\]
Similarly, we will denote the effective mass $m_{\text{eff}}(\alpha)$ that we introduced earlier by $M(\alpha)$ in accordance with \cite{LiSe13}.
The starting point is then the observation that in the case
$\delta = 0$, the Hamiltonian $H(\alpha)$
commutes with the total momentum operator $P = - i \nabla_x + P_f$ with $P_f =  \int \dd k\,  k a^\dagger(k) a(k)$. This allows a fiber decomposition  $H(\alpha) = \int_{\bbR^d}^\oplus \dd P \, H_P (\alpha)$ with  
\[
H_P(\alpha) = (P - P_f)^2 + H_{f} + \alpha \int_{\bbR^d} \de k \frac{1}{2\sqrt{\omega(k)}} \lt( \hat{\varphi}(k) a^\dagger (k) + \hat{\varphi}(-k)  a(k) \rt).
\]
See \cite{LiSe13} for more details. The operator $H_P(\alpha)$ now acts on Fock space 
only, and is bounded below. The map $P \mapsto E_\alpha(P) = \inf {\rm spec} 
H_P(\alpha)$ is radial in $P$, and by definition, the effective mass $M(\alpha)$ of the 
particle-field system is the inverse of the second derivative at $|P|=0$ of the map 
$|P| \mapsto E_\alpha(P)$. The investigation of this quantity has been an important 
open problem since Landau and Pekar \cite{LP48}
conjectured that for the Fr\"ohlich polaron, to leading order   
$M(\alpha) \sim C_{\text{LP}} \alpha^4$ as $\alpha \to \infty$, 
with $C_{\text{LP}}$ given by an explicit variational formula. This conjecture
has been completely open until rather recently, despite significant interest. This changed dramatically in the last few years, when a series of
works from different groups eventually led to a full proof of the Landau-
Pekar conjecture. See \cite{LiSe19,BP22b,Se22,BaMuSeVa24,Br24} for the 
relevant progress of the lower bound in chronological order, and 
\cite{BrSe24, Po23} for the upper bound. The results of 
\cite{Br24} and \cite{BrSe24}
now completely settle the effective mass conjecture of Landau and Pekar.

Over the years, alternative, and equivalent, definitions of the effective 
mass 
have been proposed. In his seminal paper, Spohn \cite{Sp87} showed that the 
effective mass as defined above is equivalent to two other quantities, one of 
them being the inverse of the diffusion constant of a 
(diffusively rescaled) self-interacting Brownian motion. The relevant 
probability measure is given by the infinite volume limit (on translation 
invariant events) of \eqref{nelson_path_measure} with $\delta = 0$. This 
created the bridge to stochastic processes, which has been used in 
\cite{BP22b,Se22} to investigate the effective mass. Another equivalent 
definition, which is closer to enhanced binding, is presented by Lieb and 
Seiringer in \cite{LiSe13}. They first observe that for 
a bounded, nonpositive, continuous potential $U$, the function 
\[
\caE: [1/2,\infty) \to (- \infty, 0], \quad m \mapsto \caE(m) = \inf {\rm spec}
\big( -\frac{1}{2m} \Delta + \delta U \big)
\]
is decreasing and converges to $\inf U$ as $m \to \infty$. If $-\Delta + \delta U$ has a bound state, then the function is even strictly decreasing, and one can find a unique solution $\widetilde M(\alpha)$ of the 
equation 
\[
\caE(\widetilde M(\alpha)) =  \lim_{\lambda \to 0} \frac{1}{\lambda^2} 
\inf {\rm spec} \big( - \Delta + \delta U_{\lambda} + H_{\rm f} + \alpha H_{\rm I} - E_0(\alpha)\big),
\]
where $E_0(\alpha) = \inf {\rm spec} H_{0}(\alpha)$, and where $U_\lambda(x) = \lambda^2 U(\lambda x)$. This scaling serves to make 
$U_\lambda$ locally almost constant for small $\lambda$, which intuitively 
gives a bridge to the effective mass $M(\alpha)$ defined above. The prefactor 
$\lambda^{-2}$ is necessary for a nontrivial limit, which can be seen from the case $\alpha = 0$ and  the unitary equivalence of 
$-\frac{1}{2m} \Delta + U_\lambda$ and $\lambda^2 (-\frac{1}{2m} \Delta + U)$.

The main result of \cite{LiSe13} is that $M(\alpha) = \widetilde M(\alpha)$ 
under the condition that $- \Delta + \delta U$ has a bound state, and some 
further mild conditions. In their outlook section, the authors state that 
their techniques can be useful for enhanced binding, i.e.\ for the situation 
where $- \Delta + \delta U$ does not have a bound state. No details are 
given, so here is our guess what they might have had in mind. Equation (8) of \cite{LiSe13} 
states that 
\[
\lim_{\lambda \to 0} \frac{1}{\lambda^2} \big \langle \Psi_f,  \big( - \Delta + \delta U_{\lambda} + H_{\rm f} + \alpha H_{\rm I} - E_0(\alpha) \big) \Psi_f \big \rangle = \langle f, (-\tfrac{1}{2M(\alpha)}\Delta + \delta U) f \rangle
\]
for a class of trial states $\Psi_f$ depending on functions $f$, where $f$ can be picked from a dense subset of $L^2(\bbR^d)$. By minimizing over $f$ and using that 
$M(\alpha)$ diverges as $\alpha \to \infty$ (which was yet unknown at the 
time when \cite{LiSe13} was written), one can make the right hand side 
strictly negative for given $\delta$ by taking $\alpha$ large enough. This 
means that there is some finite $\lambda$ so that the left hand side is 
strictly negative, which by the Rayleigh-Ritz principle implies that 
$\inf{\rm spec} H < 0 = \inf{\rm spec} H_{\rm f} + \inf {\rm spec} H_{\rm p}$.
This allows the use of the binding condition. In other words, this method 
shows enhanced binding for sufficiently large coupling $\alpha$ and 
sufficiently scaled potential $U_\lambda$. What it does not show, to our 
knowledge, is enhanced binding for a fixed, given potential $U$, even when 
taking $\alpha$ very large. Thus for this case, our method seems to be the only one available at this moment. 

\subsection{Persistence of binding}
The original use of the binding condition was to prove that ground states persist under coupling to the field, i.e.\ that existence of a ground state for $\alpha  = 0$ implies existence of a ground state for $\alpha > 0$. 
In \cite{GrLiLo01}, this was shown for the Pauli-Fierz model, and in 
\cite{HiMa22}, for the Nelson model, under suitable conditions, and for all $\alpha > 0$. 
Proofs using functional integrals have in the past always suffered from a 
restriction to sufficiently small $\alpha$. This leads to various conditions on the coupling strength, such as e.g.\ in \cite{Sp98}, or in 
\cite[Theorem 2.28]{HiLo20}. The reason for this restriction is that previous approaches were unable to exploit the attractive nature of the pair interaction potential $W$, which is what the Gaussian correlation inequality enables us to do. 

Analytic approaches for showing persistence of binding 
applicable to the Nelson model include \cite{Ge00}, where it is required that $H_p$ has purely discrete spectrum and that the model is infrared regular. For a nice exposition of this approach applied to the Nelson model we refer to \cite{Hi19}.
For related work  discussing the existence of ground states in non-Fock representation which we do not focus on in this work, we refer the reader to \cite{Sasa05} and references therein.

\section{Proof of Main Results}

\subsection{Notation and the Gaussian correlation inequality}
\label{gci_intro}
In all what follows we assume \ref{potential_assumption} and \ref{pair_potential_assumption}. We start by listing all relevant path measures that will occur in 
this work. Let us recall that $\bbP$ is $d$-dimensional Wiener measure, i.e.\ the path measure of Brownian motion starting at $0$. 
Motivated by (\ref{fkn_formula}) we define the probability measures
\begin{equation}
    \label{nelson_path_measure}
    \bbPh_{\delta,\alpha,T}(\de x) \propto \e{I_0 ^T(x,\delta,\alpha)} \one_B (x_T) \bbP(\de x),
\end{equation}
where $B$ is a convex symmetric set
and
$$I_0 ^T (x,\delta,\alpha) := \delta V_0 ^T(x) +  \alpha W_0 ^T (x) :=  \delta \int_0 ^T V(x_s) \de s + \alpha \int_0 ^T \int_0 ^T W(\Vert x_t - x_s \Vert, |t-s|)\de t \de s.$$
In the display above we used $V := -U$, which gives the connection to (\ref{fkn_formula}).
Observe that, by definition, we then have 
$$\bbPh_{0,\alpha,T}(\de x) \propto \e{\alpha W_0 ^T(x)} \one_B (x_T) \bbP(\de x).$$
In the following we set $B = B(0,M)$, the closed ball centred at $0$ with radius $M$. 
Here, $M \approx \supp(V)$ will be chosen accordingly in the proof of Theorem \ref{first_recurrence_theorem}. We stress that both, Lemma \ref{recurrence_lemma} and Lemma \ref{b_exp_growth}, are not affected by the choice of $B$, as long as the set is symmetric and convex. 

In the notation just introduced we recall that order of reweighting does not matter; i.e.,
$$\bbPh_{\delta,\alpha,T}(\de x) = \frac{ 1 }{\bbE^{\bbPh_{0,\alpha,T}} \lt[ \e{\delta V_0 ^T} \rt]} \e{ \delta V_0 ^T(x)} \bbPh_{0,\alpha,T} (\de x).$$
Let us also define for $l>0$ fixed, with $T/l \in \bbN$ and $\alpha \geq 0$ 
\begin{equation}
    \label{dom_measure}
    \bbPw_{\alpha,T,l} (\de x) \propto \e{-\alpha \Tilde{W}_0 ^T (x)} \bbP(\de x), \:\:\:\:\:\:\: \Tilde{W}_0 ^T (x) :=  \sum\limits_{i=0} ^{T/l -1} \int_{li} ^{l(i+1)} \int_{li} ^{l(i+1)} \Vert x_t - x_s \Vert^2   \de t \de s.
\end{equation}
Finally, we define the restriction of a measure. Let $\mu$ be a measure on $C( [0,T] ; \bbR^d)$ and take $A \subseteq C( [0,T] ; \bbR^d)$. Then,
$$\Pi_A : \caM_1 \rightarrow \caM_1, \:\:\:\:\:\:\: \Pi_A ( \mu)(\de x) \propto \one_A (x) \mu(\de x).$$

We now recall relevant notions related to the Gaussian correlation inequality, first proven in \cite{Ro14}. In the following we fix a real, separable Banach space $X$. A symmetric function $f: X \rightarrow \bbR$ is said to be quasi-concave whenever its super-level sets are convex. (GCI) directly implies that, for $\mu$ being a centred Gaussian measure and $f_1,...,f_n$ non-negative, symmetric and quasi-concave,
\begin{equation}
    \label{gci}
    \mu \lt( \prod\limits_{i=1} ^m f_i \rt) \mu \lt( \prod\limits_{j=m+1} ^n f_j \rt) \leq \mu \lt( \prod\limits_{i=1} ^{m+n} f_i \rt).
\end{equation}
We will use the following version of (GCI) frequently.
\begin{theorem}
    Let $\mu$ be a centred Gaussian measure on $X$. Define
    $$\mu^{(f)}(\de x) \propto \e{f(x)} \mu(\de x).$$
    If $f + g$ is symmetric and quasi-concave (or the uniformly bounded product of such functions) and such that $\mu^{(-g)}$ remains a centred Gaussian measure, then
    $$\mu^{(f)}(A) \ge \mu^{(-g)}(A)$$
    for all $A$ symmetric and convex.
\end{theorem}
\begin{proof}
    Use (GCI) on $\mu^{(-g)}$. Specifically, it holds that
    $\de \mu^{(f)}/\de \mu^{(-g)}$
    is symmetric and quasi-concave, so 
    $$\mu^{(-g)}(A) = \int \mu^{(-g)}(\de x) \one_A \int \mu^{(-g)}(\de x) \frac{\de \mu^{(f)}}{\de \mu^{(-g)}}(x) \le \int \mu^{(f)} \one_A = \mu^{(f)}(A).\qedhere$$
\end{proof} 
For the remainder of this section we set $X = C([0,T]; \bbR^d)$.

\begin{proposition}
    \label{conditioning_prop}
    Let $f \ge 0$ be symmetric and quasi-concave. Moreover, let $A \subseteq X$ be a convex and symmetric set. Then
    $$\int \Pi_A(\bbPw_{\alpha,T,l}) (\de x) f(x) \geq \int \bbPw_{\alpha,T,l} (\de x) f(x).$$
\end{proposition}
\begin{proof}
This is just (GCI) as introduced above.
\end{proof}

Let us for now assume \ref{local_decrease}. We define
$$A_i (R,l):= \lt\{ x \in C( [0,T]; \bbR^d) : \sup\limits_{il \leq s,t \leq l(i+1)} \Vert x_t  - x_{s} \Vert \leq R \rt\}$$
and write
$$A(R,l) = \bigcap\limits_{i=0} ^{T/l -1} A_i(R,l),$$ 
assuming $T/l \in \N$.  By \ref{local_decrease} for all $l^* \ge l >0$ and some $R > 0$ a corresponding  $\beta = \alpha / \xi > 0$ exists such that
$$\bbR^d \ni x_t - x_s \mapsto \alpha W(\Vert x_t - x_s \Vert, |t-s|) + \beta  \Vert x_t - x_s \Vert^2$$
is decreasing on $A(R,l)$ with $|t-s|\le l$ being in the same block (meaning $t,s \in [il,(i+1)l]$ for some $i \in T/l -1$).

\begin{proposition}
    \label{conditioning_prop_2}
    Take $R,l$ and $\beta$ as just introduced. Then, for every non-negative (or bounded) symmetric and quasi-concave function $f$
    $$\int \Pi_{A(R,l)}(\bbPh_{\delta,\alpha,T}) (\de x) f(x) \geq \int \bbPw_{\beta,T,l} (\de x) f(x).$$
\end{proposition}
\begin{proof}
    Note that
    $$\frac{\de \Pi_{A(R,l)}(\bbPh_{\delta,\alpha,T}) }{ \de \bbPw_{\beta,T,l}} = \one_{A(R,l)} \e{ \delta V_0 ^T + \alpha W_0 ^T + \beta \Tilde{W}_0 ^T}.$$
By construction we find that, for $s,t$ being in the same block
$$x_t - x_s \mapsto \alpha W(\Vert x_t - x_s \Vert, |t-s|) + \beta  \Vert x_t - x_s \Vert^2$$
is a quasi-concave function on $A(R,l)$. This implies via approximation by Riemann sums that
$\e{ \alpha W_0 ^T + \beta \Tilde{W}_0 ^T}$
is the limit of products of quasi-concave functions. 
The result follows via (GCI), since Gaussian measures reweighted by quadratic forms remain Gaussian. 
\end{proof}

\subsection{A preliminary confinement result}
\label{first_confinement}
Assumption \ref{potential_assumption} allows us to write
\begin{equation}
    \label{particle_potential_decomp}
     V =  V_w + V_r ,
\end{equation}
where $V_w(x) = \one_{[0,r]}(\Vert x \Vert)$ with $r>0$ chosen correspondingly and $V_r$ is a non-positive remainder term vanishing at $\infty$. 
\begin{proposition}
    Fix $R>0$. Then,
    $$\bbP \lt( \sup\limits_{t \le l} \Vert x_t \Vert \le R \rt)^{1/l} \stackrel{l \rightarrow 0}{\longrightarrow } 1.$$
\end{proposition}
\begin{proof}
    Recall that
    \[
    \bbP \lt(  \Vert x_l \Vert \ge R \rt) \le d \e{- \frac{R^2}{2 d^2l^2}}.
    \]
    Note that for $c\in (0,1)$ and $1/l\geq 1$, concavity of the logarithm implies 
    $(1-c)^{1/l}\geq 1-\frac{c}{l}$ since $(1/l)\log(1-c)\geq \log(1-\frac{c}{l})$. 
    Combining finishes the proof.
\end{proof}
In the following we denote by 
$\delta^*$ the critical value from Theorem \ref{potential_well_parameter_existence} for the potential well $V_w$ (as defined in (\ref{particle_potential_decomp})). Moreover, take $R >0$ suitable for $W$ in the sense of assumption \ref{local_decrease}.
\begin{lemma}
    \label{recurrence_lemma}
    Assume  \ref{local_decrease}.
    Fix $\delta > 0$. Take $0 < l \le l^*$ small enough such that
\begin{equation}
    \label{l_requirement}
    C_{\text{BM}}(R) :=\bbP(A_0 (R,l)) > 0.9, \:\:\:\:\:\:\:\:\: C_{BM}(R)^{1/l} \geq \exp(-\delta/8).
\end{equation}
    Then, with  $\alpha^*(\delta) = O \lt(  \delta^{-2} \rt)$ we have for all $\alpha \ge \alpha^*$ and $T$ sufficiently large
$$\bbE^{\bbPh_{0,\alpha,T}} \lt[ \exp\lt( \delta V_0 ^T \rt) \rt] \geq \exp \lt(\frac{\delta}{4}  T \rt).$$
\end{lemma}

\begin{proof}
Before going into the technicalities we want to briefly explain the idea of the proof. By our assumptions on $W$ it holds that for some $C>0$ 
$$\bbE^{\bbPh_{0,\alpha,T}} \lt[ \exp\lt( \delta V_0 ^T \rt) \rt] \geq C ^{T} \bbE^{\bbPw_{\beta,T,l}} \lt[ \exp\lt( \delta V_0 ^T \rt)\rt].$$
Since the expectation on the R.H.S. is w.r.t. a Gaussian measure which diffuses slower than Brownian motion (on fixed time scales), it is reasonable to expect that
\begin{equation}
    \nn
    \begin{split}
        \bbE^{\bbPw_{\beta,T,l}} \lt[ \exp\lt( \delta V_0 ^T \rt)\rt] = \bbE^{\bbPw_{\beta,T,l}} \lt[ \exp\lt( \delta \int_0 ^T \de s V(x_s) \rt)\rt]  &\ge \bbE^{\bbP} \lt[ \exp\lt( \delta \int_0 ^T \de s V(x_s / \Tilde{v}) \rt)\rt] \\
        &= \bbE \lt[ \exp\lt( \delta \Tilde{v}^2 \int_0 ^{T/\Tilde{v}^2} \de s V(x_s) \rt)\rt],
    \end{split}
\end{equation}
where $\Tilde{v} = \Tilde{v}(\alpha) \stackrel{\alpha \rightarrow \infty}{\longrightarrow}\infty$. The expectation on the R.H.S. now corresponds to the Schrödinger operator $ - \frac{\Delta}{2} + \delta \Tilde{v}^2 V,$ for which the existence of a ground state (and so exponential growth of the expectation) is given by assumption for $\Tilde{v}$ sufficiently large.

To make this argument rigorous we abbreviate $\one_{[0,r]}(\Vert z \Vert) =\one_{r}(\Vert z \Vert)$ and start by showing
$$\bbE^{\bbPh_{0,\alpha,T}} \lt[ \exp\lt( \delta V_0 ^T \rt) \rt] \geq C ^{T} \bbE^{\bbPw_{\beta,T,l}} \lt[ \prod\limits_{i=0} ^{T/l -1} \exp \lt( l\delta \one_{ r/2}(\Vert x_{li} \Vert) \rt)\rt]$$
    with
    \begin{equation}
\label{weak_recurrence_constant}
C = C(\alpha,R) = C_{\text{BM}}(R)^{1/l} (1- (l^{-3}\beta)^{-100})^{1/l}.
\end{equation}

    W.l.o.g. we can assume that $V = \one_{r}(\Vert \cdot \Vert )$  since we can decompose $V$ as shown in (\ref{particle_potential_decomp}).
    As before, we take $\beta := \alpha/\xi$ as given in \ref{local_decrease}.
    It follows from Proposition \ref{conditioning_prop_2} that
    \begin{equation}
        \nn
        \begin{split}
            \bbE^{\bbPh_{0,\alpha,T}} \lt[ \exp\lt( \delta V_0 ^T \rt) \rt] &\geq \bbPh_{0,\alpha,T}(A(R,l)) \bbE^{\Pi_{A(R,l)}(\bbPh_{0,\alpha,T})} \lt[ \exp\lt( \delta V_0 ^T \rt) \rt] \\
            &\geq C_{\text{BM}}(R)^{T/l} \bbE^{\bbPw_{\beta,T,l}} \lt[ \exp\lt( \delta V_0 ^T \rt) \rt].
        \end{split}
    \end{equation}
    Note that this estimate is valid for any $B$ in the definition of $\bbPh_{0,\alpha,T}$, as long as the set is symmetric and convex.
    Proposition \ref{conditioning_prop} then already suffices to deduce the first estimate. Indeed, recalling that  $r$ denotes the width of the potential well, we find with 
    $C_1 = \bbPw_{\beta,T,l} \lt( A_0(r/2,l)  \rt)$ that
    \begin{equation}
        \nonumber
        \begin{split}
        \bbE^{\bbPw_{\beta,T,l}} \lt[ \exp\lt( \delta V_0 ^T \rt) \rt] &\geq C_1 ^{T/l} \bbE^{\Pi_{A(r/2,l)}(\bbPw_{\beta,T,l})} \lt[ \exp \lt( \delta V_0 ^T \rt)\rt] \\ 
        &\geq C_1 ^{T/l}  \bbE^{\Pi_{A(r/2,l)}(\bbPw_{\beta,T,l})} \lt[ \prod\limits_{i=0} ^{T/l -1} \exp \lt( l\delta \one_{r/2}(\Vert x_{li} \Vert) \rt)\rt] \\
            &\geq C_1 ^{T/l}  \bbE^{\bbPw_{\beta,T,l}} \lt[ \prod\limits_{i=0} ^{T/l -1} \exp \lt( l\delta \one_{r/2} (\Vert x_{li} \Vert)\rt)\rt].
        \end{split}
    \end{equation}
The estimate 
$C_1 \ge (1 -  (l^{-3} \beta)^{-100})$
follows from a change of variables, Brownian rescaling and \cite[Lemma 4.4]{Se22} for $\alpha$ large enough.

The discretization that was introduced in the exponent allows us to extract the effective binding of $\bbPw_{\beta,T,l}$. 
Due to the underlying structure of this measure we can interpret $(x_{ln})_{n \in \N}$ as a Gaussian random walk where the one-step increments have variance $v^2$. Let us now show
\begin{equation}
    \nonumber
    \bbE^{\bbPw_{\beta,T,l}} \lt[ \prod\limits_{i=0} ^{T/l - 1} \exp \lt( l\delta \one_{r/2}(\Vert x_{il} \Vert)\rt)\rt] \geq C_P ^{T/l} \bbE^{\bbP} \lt[ \exp \lt( \Tilde{\delta} \int_0 ^{16v^2T/l } \de s V(x_s) \rt)\rt],
\end{equation}
where $\Tilde{\delta} = \frac{l\delta }{16 v^2}$ and $C_P = C_P (\beta) = \bbP ( A_0 ( r/4v,1 )).$ We will see shortly that $C_P \uparrow 1$ as $\beta \uparrow \infty$. 
Define now
$$\Tilde{A} := \bigcap\limits_{i=0}^{T/l-1} A_i (r/4v,1)$$
to find
\begin{equation}
    \nonumber
    \begin{split}
        \bbE^{\bbPw_{\beta,T,l}} \lt[ \prod\limits_{i=0} ^{T/l-1} \exp \lt( l\delta \one_{r/2}(\Vert x_{li} \Vert) \rt)\rt] &= \bbE^{\bbP} \lt[ \prod\limits_{i=0} ^{T/l-1} \exp \lt( l\delta \one_{ r/2v}(\Vert x_i \Vert) \rt)\rt] \\
        &\ge C_P ^{T/l} \bbE^{\Pi_{\Tilde{A}}(\bbP)} \lt[ \prod\limits_{i=0} ^{T/l - 1} \exp \lt( l\delta \one_{ r/2v}(\Vert x_i \Vert) \rt)\rt].
    \end{split}
\end{equation}
On $\Tilde{A}$ we have that $ \one_{r/2v }(\Vert x_i \Vert) \geq \int_{i} ^{i+1} \de s \one_{r/4v}(\Vert x_s \Vert)$ for every $i \in \bbN$ and so, using that $\frac{\de \Pi_{\Tilde{A}}(\bbP) }{ \de \bbP}$ is symmetric and quasi-concave,
\begin{equation}
    \nonumber
    \begin{split}
        C_P ^{T/l} \bbE^{\Pi_{\Tilde{A}}(\bbP)} \lt[ \prod\limits_{i=0} ^{T/l -1} \exp \lt( l \delta \one_{r/2v}(\Vert x_i \Vert) \rt)\rt] &\geq C_P ^{T/l} \bbE^{\Pi_{\Tilde{A}}(\bbP)} \lt[ \exp \lt( l\delta \int_0 ^{T/l} \de s \one_{r/4v}(\Vert x_s \Vert) \rt)\rt] \\
        &\geq C_P ^{T/l} \bbE^{\bbP} \lt[ \exp \lt( l\delta \int_0 ^{T/l} \de s \one_{r/4v}(\Vert x_s \Vert) \rt)\rt] \\
        &= C_P ^{T/l} \bbE^{\bbP} \lt[ \exp \lt( \Tilde{\delta} \int_0 ^{16v^2 T/l} \de s V(x_s) \rt)\rt].
    \end{split}
\end{equation}
We recall from \cite[Section 4]{Se22} that $v^2$ is $O(\alpha^{-1/2})$ under $\bbPw_{\beta,T,l}$.
It is therefore possible to choose $\alpha$ large enough such that $\Tilde{\delta}$ is bigger than the constant $\delta^*$ obtained from Theorem \ref{potential_well_parameter_existence} for the potential well $V$. 
Moreover, this shows that the required order of coupling strength required for $\delta >0$ is $\delta^{-2}$.
It then follows from Corollary \ref{ground_state_exponential_growth} that
$$\bbE^{\bbP} \lt[ \exp \lt( \Tilde{\delta} \int_0 ^{16v^2 T/l} \de s V(x_s) \rt)\rt] \geq \exp \lt(  8 \Tilde{\delta} v^2 T/l \rt) = \exp \lt(  \frac{\delta}{2} T\rt) $$
for $T$ sufficiently big.
Now, if necessary, increase $\alpha$ so that \footnote{Note that there is a $v$ in the definition of $C_P$, allowing us to make the compact set as large as desired.} the entire pre-factor satisfies $C C_P \geq \exp(-\delta/4)$ (recall our choice for $R$), which shows the claim.
\end{proof}
The additional binding of the field is not necessary in case $H_p$ has a spectral gap. Indeed, the field only enhances binding by an unknown quantity that we can ignore thanks to (GCI).
\begin{lemma}
    \label{b_exp_growth}
    Assume $H_p$ has a spectral gap. Then, for all $\alpha \ge 0$ there exists $c>0$ such that for all $T$ sufficiently large
    $$\bbE^{\bbPh_{0,\alpha,T}} \lt[ \exp\lt( \delta V_0 ^T \rt) \rt] \geq \exp(cT).$$
\end{lemma}
\begin{proof}
    Use (GCI) to estimate
    $$\bbE^{\bbPh_{0,\alpha,T}} \lt[ \exp\lt( \delta V_0 ^T \rt) \rt] \ge \bbE^{\bbP} \lt[ \exp\lt( \delta V_0 ^T \rt) \rt].$$
    By assumption $H_p$ has a ground state with negative ground state energy, which implies the result (follow the proof of Theorem \ref{ground_state_exponential_growth} to deduce this).
\end{proof}

In the next theorem, we provide sufficient conditions to imply that
\begin{equation}
    \label{p_recurrence_property}
    \bbPh_{\delta,\alpha,T} \lt( \int_0 ^T \one_K (x_s) \de s > cT \rt) \geq 1- \exp(-cT).
\end{equation}

\begin{theorem}
\label{first_recurrence_theorem}
    Fix $\delta > 0$.
    \begin{enumerate}
        \item 
        \label{p_recurrence_first_item}
        Assume \ref{local_decrease}. There exists $\alpha^* \geq 2$, $c> 0$ and a compact set $K \subseteq \bbR^d$ such that for all $\alpha \ge \alpha^*$ and $T$ sufficiently large (\ref{p_recurrence_property}) holds. 
        \item 
        \label{p_recurrence_first_item_2}
        Assume $H_p$ has a spectral gap. There exists $c>0$ and a compact set $K \subseteq \bbR^d$ such that for all $\alpha \ge 0$ and $T$ sufficiently large (\ref{p_recurrence_property}) holds. 
    \end{enumerate}
\end{theorem}

\begin{proof}
    We provide a proof for (\ref{p_recurrence_first_item}) and indicate the statements that have to be changed for (\ref{p_recurrence_first_item_2}) in parentheses. Let us write
    $$V = V \one_{[0,M]} + V \one_{[0,M]^c} =: \Tilde{V}+ \overline{V},$$
    where $M$ is chosen such that $V(x) \leq \epsilon$ whenever $\Vert x \Vert \geq M$ (recall that our assumptions justify $\lim\limits_{\Vert x \Vert \rightarrow \infty}V(x) = 0$). Lemma \ref{recurrence_lemma} (or Lemma \ref{b_exp_growth}) yields a constant $c>0$ such that
    $$\bbE^{\bbPh_{0,\alpha,T}} \lt[ \e{ \delta V_0 ^T} \rt] \geq \exp(cT)$$ 
    for $\alpha$ sufficiently large ($\alpha \ge 0$), as well as $T$ large.
    For $c^* >0$ we can estimate 
    \begin{equation}
        \nonumber
        \begin{split}
            \bbPh_{\delta,\alpha,T} \lt( \delta \Tilde{V}_0 ^T \leq c^*T \rt) &= \frac{ 1 }{\bbE^{\bbPh_{0,\alpha,T}} \lt[ \e{\delta V_0 ^T} \rt]} \int \bbPh_{0,\alpha,T} (\de x) \e{ \delta \Tilde{V}_0 ^T + \delta \overline{V}_0 ^T} \one_{ \{ \delta \Tilde{V}_0 ^T \leq c^* T \}} \\
            &\leq e^{ (c^*-c)T}  \int \bbPh_{0,\alpha,T} (\de x) \e{ \delta \overline{V}_0 ^T} \one_{ \{ \delta \Tilde{V}_0 ^T \leq c^* T \} } \\
            &\leq e^{(c^* + \epsilon \delta - c)T}. 
        \end{split}
    \end{equation}
    Choosing $c^*,\epsilon >0$ small enough implies that $C_2 = c^* + \epsilon \delta - c$ is negative, and so, abbreviating $C_1 = c^* / \delta$
    $$\bbPh_{\delta,\alpha,T} \lt( \Tilde{V}_0 ^T (x) \geq C_1 T \rt) \geq 1-\e{C_2 T},$$
    which implies 
    $$\bbPh_{\delta,\alpha,T} \lt( \int_0 ^T \one_{[0,M]}(\Vert x_s \Vert) \de s \geq \frac{C_1}{|V(0)|} T \rt) \geq 1- \e{-|C_2| T}.$$
    Choosing the compact set $K=B(0,M)$ in (\ref{p_recurrence_property}) shows the claim.
\end{proof}

\subsection{Proof of Theorem \ref{recurrence_half_time_theorem}}
\label{section_proof_recurrence_half_time_theorem}
From now on we will assume either \ref{local_decrease} or that $H_p$ has a spectral gap. Moreover, $B = B(0,M)$ is the ball that was found in Theorem \ref{first_recurrence_theorem} and will be used in the definition of $\bbPh_{\delta,\alpha,T}$ for the indicator at the end-point.
In the following we will translate the result of Theorem \ref{first_recurrence_theorem} to the escape probabilities of the path at time $T/2$. The main goal for the remainder of this section is to derive a proof for Theorem \ref{recurrence_half_time_theorem}.
In order to do this, some estimates concerning different starting points are required. Let us recall that by the definition of a Brownian bridge it holds that, for $f_i: \bbR^d \rightarrow \bbR_+$ non-negative and $1 < t_1 \leq ... \leq t_n \leq T$
\begin{equation}
    \label{BB_representation}
    \bbE^x \Big[ \e{V_0 ^T} \!\!\! \e{W_0 ^T} \prod\limits_{i=1}^n f_i(x_{t_i}) \Big] \leq \!\!C_I \!\! \int_{\bbR^d} \de y \Pi (x-y) \bbE^{x,y} _{[0,1]} \lt[ \e{V_0 ^1}\!\!\! \e{W_0 ^1} \rt] \bbE^y  \Big[ \e{V_0 ^{T-1}} \e{W_0 ^{T-1}} \prod\limits_{i=1}^n f_i(x_{t_i - 1}) \Big].
\end{equation}
Here, $\bbE^{x,y} _{[0,1]}$ denotes the expectation w.r.t. a Brownian bridge starting in $x$, ending in $y$ of unit length and $\Pi$ is the Gaussian kernel with variance $1$. The constant $C_I > 0$ comes from estimating all $W_0 ^T$ interactions between $[0,1]$ and $[1,T]$. This is possible due to assumption \ref{pair_potential_assumption}. Similarly,
$$\int_{\bbR^d} \de y \Pi (x-y) \bbE^{x,y} _{[0,1]} \Big[ \e{V_0 ^1} \!\!\! \e{W_0 ^1} \Big] \bbE^y  \Big[ \e{V_0 ^{T-1}} \!\!\! \e{W_0 ^{T-1}} \prod\limits_{i=1}^n f_i(x_{t_i -1}) \Big] \leq C_{I} \bbE^x \Big[ \e{V_0 ^T} \e{W_0 ^T} \prod\limits_{i=1}^n f_i(x_{t_i}) \Big].$$

\begin{proposition}
\label{zero_domination}
Fix a compact set $K\subseteq \bbR^d$ and let $f_i: \bbR^d \rightarrow \bbR_+$ be radial, non-negative and $1 < t_1 \leq ... \leq t_n \leq T$. Then, there exists a constant $\overline{c} = \overline{c}(K) >0$ independent of $T$ such that
    \begin{equation}
        \nonumber
        \bbE \lt[ \e{V_0 ^T} \e{W_0 ^T} \prod\limits_{i=1}^n f_i(x_{t_i}) \rt] \leq \overline{c} \inf_{x \in K} \bbE^{x} \lt[ \e{V_0 ^T} \e{W_0 ^T} \prod\limits_{i=1}^n f_i(x_{t_i}) \rt].
    \end{equation}
\end{proposition}

\begin{proof}
    Pick a set of $N$ points $\{z_1,...,z_N\} =: M \subseteq \{ z \in \bbR^d : \Vert z \Vert = 1 \}$ so that for every $y$ with $\Vert y \Vert = 1$ there exists $z_i \in M$ with 
    $$\Vert y - z_i \Vert \leq \frac{1}{2}.$$  
    Take $R>0$ as some radius so that $K \subseteq B(0,R)$.
    Note that the map
    $$x \mapsto \bbE^x \lt[ \e{V_0 ^T} \e{W_0 ^T} \prod\limits_{i=1}^n f_i(x_{t_i}) \rt]$$
    is invariant under rotations around the origin as the law of Brownian motion is rotationally invariant and all functions that are integrated are radial.
    Furthermore,
    \[
    \sup\limits_{x,y \in K}\frac{ \Pi (y)}{\Pi(x-y)} \leq \e{2R^2}.
    \]
    Pick any $x \in K$ and set $r := \Vert x \Vert$. 
    By construction we find for $y \notin B(0,R)$ that $\Vert y \Vert \geq \min\limits_{x_i \in rM} \Vert y -  x_i \Vert$ which leads to
    \begin{equation}
        \label{sum_kernel}
        \Pi(y) \leq \max\limits_{x_i \in rM} \Pi(x_i-y) \leq \sum\limits_{x_i \in r M} \Pi ( x_i - y).
    \end{equation}
    We use the lower bound via Brownian bridges as given in (\ref{BB_representation}) and estimate the outer integral on $B(0,R)$ and its complement separately. On the first set we find, with 
    $$c := \sup\limits_{y \in K} \frac{\Pi (y) \bbE^{0,y} _{[0,1]} \lt[ \e{V_0 ^1} \e{W_0 ^1}\rt]}{\sum\limits_{x_i \in rM} \Pi (x_i-y) \bbE^{x_i,y} _{[0,1]} \lt[ \e{V_0 ^1} \e{W_0 ^1}\rt]} \le  \e{2R^2 +1 + \delta |V(0)|},$$
    \begin{equation}
        \nonumber
        \begin{split}
           &\int_{B(0,R)} \de y \Pi (y) \bbE^{0,y} _{[0,1]} \lt[ \e{V_0 ^1} \e{W_0 ^1} \rt] \bbE^y  \lt[ \e{V_0 ^{T-1}} \e{W_0 ^{T-1}} \prod\limits_{i=1}^n f_i(x_{t_i -1}) \rt] \\
           &\leq c  \int_{B(0,R)} \de y \sum\limits_{x_i \in rM} \Pi (x_i - y) \bbE^{x_i,y} _{[0,1]} \lt[ \e{V_0 ^1} \e{W_0 ^1} \rt] \bbE^y  \lt[ \e{V_0 ^{T-1}} \e{W_0 ^{T-1}}\prod\limits_{i=1}^n f_i(x_{t_i -1}) \rt] \\
           &= c N \int_{B(0,R)} \de y  \Pi (x - y) \bbE^{x,y} _{[0,1]} \lt[ \e{V_0 ^1} \e{W_0 ^1} \rt] \bbE^y  \lt[ \e{V_0 ^{T-1}} \e{W_0 ^{T-1}} \prod\limits_{i=1}^n f_i(x_{t_i-1}) \rt].
        \end{split}
    \end{equation}
    In the last line we used the radial symmetry and the fact that $|rM| = N$ by definition. On $B(0,R)^c$ we use the same radial symmetry and (\ref{sum_kernel}) to obtain the bound 
    \begin{equation}
        \nonumber
        \begin{split}
            &\int_{B(0,R)^c} \de y \Pi (y) \bbE^{0,y} _{[0,1]} \lt[ \e{V_0 ^1} \e{W_0 ^1} \rt] \bbE^y  \lt[ \e{V_0 ^{T-1}} \e{W_0 ^{T-1}} \prod\limits_{i=1}^n f_i(x_{t_i -1}) \rt] \\
            &\leq N \int_{B(0,R)^c} \de y  \Pi (x - y) \bbE^{x,y} _{[0,1]} \lt[ \e{V_0 ^1} \e{W_0 ^1} \rt] \bbE^y  \lt[ \e{V_0 ^{T-1}} \e{W_0 ^{T-1}} \prod\limits_{i=1}^n f_i(x_{t_i-1}) \rt].
        \end{split}
    \end{equation}
    Putting these two bounds together we find that, for any $x \in K$
    \begin{equation}
        \nn
        \begin{split}
            &\int_{\bbR^d} \de y \Pi (y) \bbE^{0,y} _{[0,1]} \lt[ \e{V_0 ^1} \e{W_0 ^1} \rt] \bbE^y  \lt[ \e{V_0 ^{T-1}} \e{W_0 ^{T-1}} \prod\limits_{i=1}^n f_i(x_{t_i-1}) \rt] \\
            &\leq  c N \int_{\bbR^d} \de y  \Pi (x - y) \bbE^{x,y} _{[0,1]} \lt[ \e{V_0 ^1} \e{W_0 ^1} \rt] \bbE^y  \lt[ \e{V_0 ^{T-1}} \e{W_0 ^{T-1}} \prod\limits_{i=1}^n f_i(x_{t_i-1}) \rt],
        \end{split}
    \end{equation}
    where the constants $c,N$ are independent of $x$ \footnote{And, of course, $T$. All we require is that the function we integrate against is non-negative and radially symmetric.}. Taking the inf on the R.H.S. and recalling \ref{pair_potential_assumption} shows the claim.
\end{proof}

\begin{proposition}
\label{prop_zero_max}
Let $f_i: \bbR^d \rightarrow \bbR_+$ be radial, non-negative and quasi-concave. Then, for $0 \leq t_1 \leq ... \leq t_n < \infty$
    and $\Vert x \Vert \ge \Vert y \Vert$,
    $$\bbE^x \lt[ \e{V_0 ^T} \e{W_0 ^T} \prod\limits_{i=1}^n f_i(x_{t_i}) \rt] \leq \bbE^y \lt[ \e{V_0 ^T} \e{W_0 ^T} \prod\limits_{i=1}^n f_i(x_{t_i}) \rt].$$
    In particular, for all $z \in \bbR^d$
    $$\bbE^z \lt[ \e{V_0 ^T} \e{W_0 ^T} \prod\limits_{i=1}^n f_i(x_{t_i}) \rt] \leq \bbE^0 \lt[ \e{V_0 ^T} \e{W_0 ^T} \prod\limits_{i=1}^n f_i(x_{t_i}) \rt].$$
\end{proposition}

\begin{proof}
    Take $0 \le t_1 < ... < t_n < \infty$ and $A_1,...,A_n,A_{n+1},...,A_{m}$ symmetric and convex sets. To prove the claim it is sufficient to show that for any $z \in \bbR^d$, the map
    $$t \mapsto \bbP^{t z} \big( x_{t_1} \in A_1,...,x_{t_n} \in A_n, (x_{t_n}-x_{t_1}) \in A_{n+1},... \big)$$
    is decreasing on $[0,\infty)$. Note that the vector $(x_{t_1},...,x_{t_n})$ is centred Gaussian with log-concave Lebesgue density. Moreover, all sets are symmetric and convex, so for any linear function $f: C([0,T]; \bbR^d) \rightarrow \bbR^d$ and $i \le m$, the map 
    $$ x \mapsto \one_{A_i}\big(f(x)\big),$$
    is symmetric and quasi-concave.
    The result then follows from Anderson's Theorem \cite{An55}.
    \end{proof}

We now describe the main idea to prove Theorem \ref{recurrence_half_time_theorem}: define the (non-convex) sets
\[
D_i := \{ \text{there exists $s \in [i,i+1]$ so that $x_s \in K$}\}
\]
and for $s<t$
\[
K_{[s,t]} := \lt\{ x \in C([0,T;\bbR^d]) : x_v \notin K \text{ for all } v \in [s,T/2-1] \cup [T/2 +1,t] \rt\}.
\]
We can decompose path-space as a disjoint union
$$C([0,T;\bbR^d]) = \bigcup\limits_{i= 0} ^{T/2 - 2} \bigcup\limits_{j= T/2 +1} ^{T}  \lt( D_i \cap K_{i+1,j} \cap D_j\rt),$$
where the convention $D_T = \{ x_s \notin K  \: \forall \: s \ge T/2 + 1\}$ is used.

One would expect that 
\begin{equation}
\label{eq:j-i-exponential-tail}
\bbPh_{\delta,\alpha,T} (D_i \cap K_{i+1,j} \cap D_j) \approx \bbPh_{\delta,\alpha, j-i} (K_{[1,j-i]}) \le \e{-c(j-i)},
\end{equation}
where the last estimate is due to Theorem~\ref{first_recurrence_theorem}. This is because $D_i$ guarantees the path is ``near $K$'' close to time $i$, so we should be able to re-start the process at $0$ without causing major changes, due to \eqref{eq:time-decay-integrable} and Proposition \ref{zero_domination}. Similarly for $D_j$.
Further, \eqref{eq:j-i-exponential-tail} should imply that $j-i$ has exponential tail and is in particular tight. 
Then \eqref{eq:time-decay-integrable} should allow us to approximately reduce attention to the bounded time-interval $[i,j]$, where both endpoints are at the origin.
At this point, tightness of $x_t$ should clearly hold uniformly on $t\in [i,j]$.

While \eqref{eq:j-i-exponential-tail} turns out to be true, we require some preparation to show it. Because these estimates are of technical nature, we will only record the relevant Proposition that is required to conclude the proof of Theorem \ref{main_thm} (and Proposition \ref{simple_gs}) here. A proof can be found in section \ref{approx_markov_section}.
\begin{proposition}
    \label{approx_markov}
    There exists a constant $C>0$ such that for all $T>5$ and $1 \le i < T/2 - 3$, $T-3 > j > T/2 +3$,
    \begin{equation}
        \nn
        \begin{split}
            & \bbE \lt[  \e{I_0 ^T} \one_{D_i} \one_{K_{[i+1,j]}} \one_{D_j} \one_{x_T \in K}  \rt]\\
            &\le C \bbE \lt[ \e{I_0 ^{i+1}} \one_{x_{i+1 \in K}}\rt] \bbE \lt[ \e{I_0 ^{j-i + 2}} \one_{K_{[3,j-i -1]}} \one_{x_{j-i + 2 \in K}}\rt] \bbE \lt[ \e{I_0 ^{T-j-3}} \one_{x_{T-j-3} \in K}\rt].
        \end{split}
    \end{equation} 
\end{proposition}

Let us also mention that one can always add (and remove) a constant amount of terms from the interaction term $I_0 ^T$ by multiplying with a finite constant. In other words, we always have for any $f$ for example
$$\bbE^z \lt[ \e{I_0 ^T} f(x) \rt] \le C_1 \bbE^z  \lt[ \e{I_0 ^{T-2}} f(x) \rt] \le C_2 \bbE^z \lt[ \e{I_3 ^{T+3}} f(x) \rt],$$
where the constants $C_1,C_2$ do not depend on $z \in \bbR^d$ (or $f$) and only depend on the amount of time that is added or removed (so in the above example $2$,$3$ and $5$). In the remaining part of this work we will not note this whenever such an estimate is made and only replace e.g. $C_1$ by $C_2$. We will denote by $(x_t)_{t \le T}$, $(\Tilde{x}_t)_{t \le T}$, $(\overline{x}_t)_{t \le T}$, $(x^\dag _t)_{t \le T}$ independent Brownian motions.

\begin{proof}[Proof of Theorem \ref{recurrence_half_time_theorem}]
Our goal is to find an integer $L$ so that
\begin{equation}
    \label{half_time_goal}
    \sup\limits_{T\ge 2L} \sum\limits_{i < T/2 - L, j \ge T/2 + L} \bbPh_{\delta,\alpha,T}(D_i, K_{[i+1,j]}, D_j) < \epsilon
\end{equation}
for $\epsilon > 0$.
Abbreviate
$$ D(L) := \bigcup\limits_{i= T/2 - L} ^{T/2 -2} \bigcup\limits_{j= T/2 + 1} ^{T/2 + L} D_i \cap K_{[i+1,j]} \cap D_j$$
and note that for each $x \in D(L)$ there exist $T/2 - L \le  i < T/2 -1$ and $T/2 + L \ge j > T/2 +1$ so that $x_i \in K$, $x_j \in K$ and that
\begin{equation}
    \label{non-vanishing-L}
    \liminf\limits_{T \rightarrow \infty} \bbPh_{\delta,\alpha,T}(D(L)) > 0,
\end{equation}
assuming that an $L$ for which (\ref{half_time_goal}) holds can be found. Let us therefore study the quantity $\bbPh_{\delta,\alpha,T}(D_i, K_{[i+1,j]}, D_j)$ for $i \ge 1$, $j \le T-3$ and $j-i > 4$. 
By conditioning on $\{x_{i+1} \in K, x_{j+3} \in K\}$, writing everything out and applying Proposition \ref{approx_markov} we find 
\begin{equation}
    \nn
    \begin{split}
    &\bbPh_{\delta,\alpha,T}(D_i, K_{[i+1,j]}, D_j) \\
    &= \frac{ \bbE \lt[  \e{I_0 ^T} \one_{D_i} \one_{K_{[i+1,j]}} \one_{D_j} \one_{x_T \in K } \rt] }{ \bbE \lt[ \e{I_0 ^T} \one_{x_T \in K } \rt] } \\
    &\le C \frac{ \bbE \lt[ \e{I_0 ^{i+1}} \one_{x_{i+1 \in K}}\rt] \bbE \lt[ \e{I_0 ^{j-i + 2}} \one_{K_{[3,j-i -1]}} \one_{x_{j-i + 2 \in K}}\rt] \bbE \lt[ \e{I_0 ^{T-j-3}} \one_{x_{T-j-3} \in K}\rt] }{ \bbE \lt[ \e{I_0 ^T} \one_{x_T \in K } \rt] } \\
    &\le C \frac{ \bbE \lt[ \e{I_0 ^{i+1}} \one_{x_{i+1 \in K}}\rt] \bbE \lt[ \e{I_0 ^{j-i + 2}} \one_{K_{[3,j-i -1]}} \one_{x_{j-i + 2 \in K}}\rt] \bbE \lt[ \e{I_0 ^{T-j-3}} \one_{x_{T-j-3} \in K}\rt] }{ \bbE \lt[ \e{I_0 ^T} \one_{x_{i+1} \in K} \one_{x_{j+3} \in K} \one_{x_T \in K } \rt]} =: (*).
    \end{split}
\end{equation}
We estimate the denominator by
\begin{equation}
    \nn
    \begin{split}
        & \bbE \lt[ \e{I_0 ^T} \one_{x_{i+1} \in K} \one_{x_{j+3} \in K} \one_{x_T \in K} \rt] \\
        &\ge C_1 \bbE \lt[ \e{I_0 ^{i+1}} \one_{x_{i+1} \in K} \e{I_{i+1} ^{j+3}} \one_{x_{j+3} \in K} \e{I_{j+3} ^T}\one_{x_T \in K} \rt] \\
        &\ge C_1 ^2 \overline{c}^2 \bbE \lt[ \e{I_0 ^{i+1}} \one_{x_{i+1} \in K} \rt] \: \bbE \lt[ \e{I_0 ^{j-i +2}} \one_{x_{j-i +2} \in K} \rt] \: \bbE \lt[ \e{I_0 ^{T-j-3}} \one_{x_{T-j-3} \in K} \rt].
    \end{split}
\end{equation}
The last inequality follows from Proposition \ref{zero_domination} and the Markov property of Brownian motion. Now most terms cancel, which gives
$$(*) \le C_3 \frac{ \bbE \lt[  \e{I_0 ^{j-i + 2}} \one_{K_{[3,j-i -1]}} \one_{x_{j-i + 2 \in K}} \rt]}{\bbE \lt[ \e{I_0 ^{j-i +2}} \one_{x_{j-i +2} \in K} \rt] } = C_3 \bbPh_{\delta,\alpha,j-i +2} (K_{[3,j-i -1]}) \le \e{-c(j-i+2)}$$
where the last inequality follows from (\ref{p_recurrence_property}) for $i+j$ large enough. Let us also not forget the cases $i=0$ and $j \ge T-3$. If $i =0$ we can estimate
$$\bbPh_{\delta,\alpha,T}(D_0, K_{[1,j]},D_j) \le \bbPh_{\delta,\alpha,T}( K_{[1,j]},D_j).$$
It is now possible to perform the same steps as above, with the difference that one just has a condition on $x_{j}$, which yields the same result. The case $j \ge T-3$ is treated analogously.

We have thus shown that
\[
\sum\limits_{T/2 - i \ge L, j-T/2  \ge L} \bbPh_{\delta,\alpha,T}(D_i, K_{[i+1,j]}, D_j) \le C \sum\limits_{T/2 - i \ge L, j-T/2  \ge L} \bbPh_{\delta,\alpha,j-i +2} (K_{[3,j-i -1]}),
\]
which by the above estimate can be made as small as necessary (the term on the right is independent of $T$, so the bound holds uniformly in $T$). Hence it is possible to find an $L$ such that (\ref{non-vanishing-L}) is satisfied.
We abbreviate $D = D(L)$ and see that
\[
\liminf\limits_{T \rightarrow \infty}\bbPh_{\delta,\alpha,T}(x_{T/2} \! \in \! K) \! \geq \!\liminf\limits_{T \rightarrow \infty} \bbPh_{\delta,\alpha,T} (D) \bbPh_{\delta,\alpha,T}(x_{T/2}\! \in  \!K | D) \! \geq \! C \liminf\limits_{T \rightarrow \infty} \bbPh_{\delta,\alpha,T}(x_{T/2} \! \in \! K \! | D).
\]
Using the abbreviations 
\[
D_L := \bigcup\limits_{i=T/2 - L} ^{T/2 -2} D_i;
\quad\quad
D_R :=  \bigcup\limits_{i=T/2 + 1} ^{T/2 + L} D_i 
\]
we can conclude

\begin{equation}
    \nn
    \begin{split}
        &\bbPh_{\delta,\alpha,T}(x_{T/2} \in K | D) = \frac{ \bbE \lt[ \e{I_0 ^T} \one_{D_L} \one_{x_{T/2} \in K} \one_{D_R} \one_{x_T \in K} \rt] }{ \bbE \lt[  \e{I_0 ^T} \one_{D_L} \one_{D_R} \one_{x_T \in K} \rt]} \\
        &\ge C_1 \frac{ \bbE \lt[ \e{I_0 ^{T/2 -1}} \one_{x_{T/2-1 \in K}}  \bbE^{x_{T/2 - 1}} \lt[  \e{I_0 ^{L+2}} \one_{\Tilde{x}_{1} \in K} \one_{\Tilde{x}_{L+2} \in K} \bbE^{\Tilde{x}_{L+2}} \lt[ \e{I_0 ^{T/2 - L-1}}  \one_{\overline{x}_{T/2 - L-1} \in K} \rt] \rt] \rt]  }{ \bbE \lt[  \e{I_0 ^T} \one_{D_L} \one_{D_R} \one_{x_T \in K} \rt]} \\
        &\ge C_1 \overline{c}^2 \frac{ \bbE \lt[ \e{I_0 ^{T/2 -1}} \one_{x_{T/2-1 \in K}}\rt]  \bbE \lt[  \e{I_0 ^{L+2}} \one_{\Tilde{x}_{1} \in K} \one_{\Tilde{x}_{L+2} \in K} \rt] \bbE \lt[ \e{I_0 ^{T/2 - L -1}}  \one_{\overline{x}_{T/2 - L-1} \in K}   \rt]  }{ \bbE \lt[  \e{I_0 ^T} \one_{D_L} \one_{D_R} \one_{x_T \in K} \rt] } \\
        &\ge C_2 \overline{c}^2 \frac{ \bbE \lt[ \e{I_0 ^{T/2 -1}} \one_{x_{T/2-1 \in K}}\rt]  \bbE \lt[  \e{I_0 ^{L+2}} \one_{\Tilde{x}_{1} \in K} \one_{\Tilde{x}_{L+2} \in K} \rt] \bbE \lt[ \e{I_0 ^{T/2 - L -1}}  \one_{\overline{x}_{T/2 - L-1} \in K}   \rt]  }{ \bbE \lt[ \e{I_0 ^{T/2 -1}} \one_{x_{T/2-1 \in K}}\rt]  \bbE \lt[  \e{I_0 ^{L+2}} \one_{\Tilde{x}_{L+2} \in K} \rt] \bbE \lt[ \e{I_0 ^{T/2 - L -1}}  \one_{\overline{x}_{T/2 - L-1} \in K}   \rt] } \\
        &\ge C_2 \overline{c} \bbP \lt( x_1 \in K \rt).
    \end{split}
\end{equation}
The first inequality follows from modifying the interaction terms and replacing $\one_{D_L} \ge \one_{x_{T/2-1} \in K}$ (similarly for $D_R$).The next inequality follows from Proposition \ref{zero_domination} (modify the proof to also allow $t=1$). The last estimate is (GCI), since the entire density is now symmetric and quasi-concave. To obtain the estimate for the denominator we calculate 
\begin{equation}
    \nn
    \begin{split}
        &\bbE \lt[  \e{I_0 ^T} \one_{D_L} \one_{D_R} \one_{x_T \in K} \rt] \\
        &\le C_1 \bbE \lt[  \e{I_0 ^{T/2 - 1}} \one_{D_L} \bbE^{x_{T/2 -1}} \lt[  \e{I_0 ^{L+2}}\one_{\Tilde{D}_R} \bbE^{\Tilde{x}_{L+2}} \lt[ \e{I_0 ^{T/2 - L -1}} \one_{\overline{x}_{T/2 -L -1} \in K} \rt] \rt] \rt] \\
        &\le C_1 \bbE \lt[  \e{I_0 ^{T/2 - 1}} \one_{D_L} \bbE^{x_{T/2 -1}} \lt[  \e{I_0 ^{L+2}}\one_{\Tilde{D}_R}  \rt] \rt] \bbE \lt[ \e{I_0 ^{T/2 - L -1}} \one_{\overline{x}_{T/2 -L -1} \in K} \rt]   \\
        &\le C_2 \bbE \lt[  \e{I_0 ^{T/2 - 1}} \one_{D_L} \rt]\bbE \lt[  \e{I_0 ^{L+2}}\one_{\Tilde{x}_{L+2} \in K}  \rt]  \bbE \lt[ \e{I_0 ^{T/2 - L -1}} \one_{\overline{x}_{T/2 -L -1} \in K} \rt] \\
        &\le C_3 \bbE \lt[  \e{I_0 ^{T/2 - 1}} \one_{x_{T/2 -1} \in K} \rt]\bbE \lt[  \e{I_0 ^{L+2}}\one_{\Tilde{x}_{L+2} \in K}  \rt]  \bbE \lt[ \e{I_0 ^{T/2 - L -1}} \one_{\overline{x}_{T/2 -L -1} \in K} \rt].
    \end{split}
\end{equation}
Here, the first inequality again follows from modifying the interaction terms. The second estimate follows from Proposition \ref{prop_zero_max}, and the third  
  from a calculation analogous to (\ref{existence_vs_endpoint_pin}) \footnote{Note that $L$ is now a fixed number, so it is possible to introduce/remove all interaction terms in the second term.}. The same goes for the last inequality, which then shows the claim.
\end{proof}

\subsection{Proof of Theorem \ref{main_thm} and Proposition \ref{simple_gs}}
\label{main_thm_proofs}
We now provide the proof for the main Theorem of this article. The proof for Proposition \ref{simple_gs} is completely analogous.
\begin{proof}[Proof of Theorem \ref{main_thm}]
Let $\psi$ be the indicator of the symmetric, convex and compact set $K$ that is given by Theorem \ref{recurrence_half_time_theorem}.
We will show
\begin{equation}
\label{eq:liminf-rayleigh-non-zero}
\liminf\limits_{T \rightarrow \infty} \frac{\la \psi, \e{-TH} \psi \ra}{ \Vert \e{-TH} \psi \Vert_2} > 0.
\end{equation}
It is standard that this implies the desired claim; see Theorem \ref{ground_state_options} for details.

We first prove that
\begin{equation}
    \label{final_proof_first_bound}
    \frac{  \la \e{-T H}\psi, \psi \e{-T H}\psi \ra  }{ \Vert \e{-T H}\psi \Vert_2 ^2} = \frac{  \la \e{-T H}\psi, \psi \e{-T H}\psi \ra  }{ \la \psi, \e{-2T H}\psi \ra  } \geq \Tilde{C} > 0.
\end{equation}
We recall that we can choose w.l.o.g. this same set for the conditioning of $\bbPh_{\delta,\alpha,2T}$.
The denominator of (\ref{final_proof_first_bound}) can be related to the corresponding partition function of $\bbPh_{\delta,\alpha,2T}$ by (using $C_m = \lambda(K)$ \footnote{Here $\lambda$ denotes the Lebesgue measure.})
\begin{equation}
    \nonumber
    \begin{split}
        \la \psi, \e{-2T H}\psi \ra \leq C_m C_I \bbE \lt[ \e{I_0 ^{2T}} \psi(x_T)\rt]
    \end{split}
\end{equation}
For the nominator Proposition \ref{zero_domination} yields a constant $C_I$ so that
\begin{equation}
    \nonumber
    \begin{split}
        \la \e{-T H}\psi, \psi \e{-T H}\psi \ra &= \int \de x \psi (x) \bbE^x \lt[ \e{I_0 ^{T} }\psi(x_{T}) \bbE^{x_{T}} \lt[ \e{I_0 ^{T}} \psi(\Tilde{x}_{T}) \rt]\rt] \\
        &\geq C_m C_I \bbE \lt[ \e{I_0 ^{T} }\psi(x_{T}) \bbE^{x_{T}} \lt[ \e{I_0 ^{T}} \psi(\Tilde{x}_{T}) \rt]\rt] \\
        &\geq C_m C_I \bbE \lt[ \e{I_0 ^{2T} }\psi(x_{T}) \psi(x_{2T}) \rt].
    \end{split}
\end{equation}
But now, the expression in (\ref{final_proof_first_bound}) is proportional to 
$\bbPh_{\delta,\alpha,2T} (x_{T} \in K)$,
and so the estimate follows from Theorem \ref{recurrence_half_time_theorem}.

Rewriting (\ref{final_proof_first_bound}) yields
\begin{equation}
    \label{quadratic_non_vanish}
    \int \de x  \frac{\lt(\e{-TH}\psi\rt)(x)^2}{ \Vert \e{-TH}\psi \Vert_2 ^2} \psi(x) \geq \Tilde{C} > 0.
\end{equation}
Using starting point comparisons once again, we find that
\[
x \mapsto \frac{\lt(\e{-TH}\psi(x)\rt)^2}{ \Vert \e{-TH}\psi \Vert_2 ^2}
\]
is uniformly bounded in $T$ with a maximum at $0$. In conjunction with (\ref{quadratic_non_vanish}) this gives
\[
\liminf\limits_{T \rightarrow \infty} \int \de x \frac{ \lt( \e{-TH}\psi  \rt)(x)}{\Vert \e{-TH}\psi \Vert_2 } \psi(x) > 0.
\]
This is equivalent to \eqref{eq:liminf-rayleigh-non-zero}, so as previously mentioned Theorem \ref{ground_state_options} finishes the proof.
\end{proof}

\section{Proof of Proposition \ref{approx_markov}}
 \label{approx_markov_section}
 \setcounter{equation}{2}
\begin{proposition}
    \label{move_pinning_to_left_prop}
    There exists a constant $C>0$ such that for all $T > 2$ it holds that
    $$\bbE \lt[ \e{I_0 ^{T-1}} \one_{x_{T-1} \in K}  \rt] \ge C \bbE \lt[ \e{I_0 ^{T}} \one_{x_{T} \in K}  \rt].$$
\end{proposition}

\begin{proof}
    This is a direct application of Proposition \ref{prop_zero_max}. Indeed,
    \begin{equation}
        \nn
        \begin{split}
            \bbE \lt[ \e{I_0 ^{T}} \one_{x_{T} \in K}\rt] \le C \bbE \lt[ \e{I_1 ^{T}} \one_{x_{T} \in K}\rt] &= C\int_{\bbR^d}\de y \Pi(y) \bbE^{y} \lt[\e{I_0 ^{T-1}} \one_{x_{T-1} \in K} \rt] \\
            &\le C\int_{\bbR^d}\de y \Pi(y) \bbE \lt[\e{I_0 ^{T-1}} \one_{x_{T-1}\in K} \rt] \\
            &=C \bbE \lt[\e{I_0 ^{T-1}} \one_{x_{T-1}\in K} \rt]. \qedhere
        \end{split}
    \end{equation}
\end{proof}

\begin{proposition}
\label{pin_move_prop}
    Let $X \sim \caN(0,\text{I}_{d \times d})$.
    Then, for all $\sigma \in [1,2]$, for every $R > 0$ and $z \in \bbR^d$
    $$\bbP ( \Vert X + z \Vert \le R) \le \sqrt{2^d} \bbP \lt( \Vert \sqrt{\sigma} X + z \Vert \le R \rt).$$
\end{proposition}
\begin{proof}
    By the convolution formula it suffices to show the claim in one dimension, i.e. we show for every $r \in \bbR$
    $$\bbP (  | Y - r | \le R ) \le C \bbP (  | \sqrt{\sigma}Y - r | \le R )$$
    with $Y \sim \caN(0,1)$. The result is then a direct calculation. Indeed,
    \begin{equation}
        \nn
        \begin{split}
            \bbP (  | Y - r | \le R ) = \bbP( Y \in [y-R,y+R]) &= \frac{1}{\sqrt{2 \pi}} \int_{r - R} ^{r+R} \de x \e{-\frac{x^2}{2}} \\
            &\leq \sqrt{2} \frac{1}{\sqrt{2 \pi \sigma}} \int_{y-R} ^{y+R} \e{- \frac{x^2}{2}} \de x \\
            &\le \sqrt{2}\frac{1}{\sqrt{2 \pi \sigma}} \int_{y-R} ^{y+R} \e{- \frac{x^2}{2 \sigma}} \de x \\
            &= \sqrt{2} \bbP( \sqrt{\sigma} Y \in [y-R,y+R]).
        \end{split}
    \end{equation}
    Here, the last inequality simply follows from
    $$- \frac{x^2}{2} \le - \frac{x^2}{2 \sigma}$$
    for every $x \in \bbR$ as $\sigma \ge 1$.
\end{proof}
 
\begin{proposition}
\label{technical_estimate_thm1_4}
    Fix a compact set $K \subseteq \bbR^d$. There exists a constant $C>0$ such that for all $y \in K$ and $m > 5$
    $$ \bbE^{y} \lt[   \e{I_0 ^{m}} \one_{K_{[2,m-2]}}  \one_{x_{m} \in K} \rt] \le C  \bbE \lt[   \e{I_0 ^{m+1}} \one_{K_{[3,m-2]}}  \one_{x_{m+1} \in K}\rt].$$
\end{proposition}
\begin{proof}
    Fix $y \in K$ and set $r := \Vert y \Vert$. We denote by $\sigma_r (x)$ the first exit time of $B(0,r)$.
    An application of the strong Markov property yields
    \begin{equation}
        \nn
        \begin{split}
            \bbE \lt[  \e{I_0 ^{m}} \one_{K_{[3,m-3]}}  \one_{x_{m} \in K} \rt]&\ge \bbE \lt[ \e{I_0 ^{m}} \one_{\sigma_r < 1} \one_{K_{[3,m-3]}}  \one_{x_{m} \in K} \rt] \\
            &\ge C_1 \bbE \lt[ \e{I_1 ^{m}} \one_{\sigma_r < 1} \one_{K_{[3,m-3]} }  \one_{x_{m} \in K}\rt] \\
            &= C_1 \bbE \lt[ \one_{\sigma_r < 1} \bbE^{x_{\sigma_r}} \lt[ \e{I_{1-\sigma_r} ^{m-\sigma_r}} \one_{K_{[3-\sigma_r,m-3-\sigma_r]}} \one_{\Tilde{x}_{m-\sigma_r} \in K}  \rt] \rt] \\
            &\ge C_2 \bbE \lt[ \one_{\sigma_r < 1} \bbE^{x_{\sigma_r}} \lt[ \e{I_0 ^{m-2}} \one_{K_{[2,m-3]}} \one_{\Tilde{x}_{m-\sigma_r} \in K}  \rt] \rt] \\
            &= C_2 \bbE \lt[ \one_{\sigma_r < 1} \bbE^{x_{\sigma_r}} \lt[ \e{I_0 ^{m-2}} \one_{K_{[2,m-3]} } \bbE^{\Tilde{x}_{m-2}} \lt[ \one_{\overline{x}_{2- \sigma_r} \in K} \rt]   \rt] \rt] =: (*).
        \end{split}
    \end{equation}
    By Proposition \ref{pin_move_prop} we can estimate, taking $Y \sim \caN(0,1)$
    \begin{equation}
        \nn
        \begin{split}
            \bbE^{\Tilde{x}_{m-2}} \lt[ \one_{\overline{x}_{2- \sigma_r} \in K} \rt] &= \bbP ( \Vert \sqrt{2-\sigma_r}Y + \Tilde{x}_{m-2} \Vert \le R) \\
            &\ge \frac{1}{\sqrt{2^d}} \bbP ( \Vert Y + \Tilde{x}_{m-2} \Vert \le R) \\
            &= \frac{1}{\sqrt{2^d}} \bbE^{\Tilde{x}_{m-2}} \lt[ \one_{\overline{x}_1 \in K} \rt].
        \end{split}
    \end{equation}
    By rotational symmetry we find
    $$\bbE^{x_{\sigma_r}} \lt[ \e{I_0 ^{m-1}} \one_{K_{[1,m-3]}} \one_{\Tilde{x}_{m-1} \in K}  \rt] = \bbE^{y} \lt[ \e{I_0 ^{m-1}} \one_{K_{[1,m-3]}} \one_{\Tilde{x}_{m-1} \in K}  \rt],$$
    which yields the final estimate
    \begin{equation}
        \nn
        \begin{split}
            (*) &\ge \frac{C_2}{\sqrt{2^d}}  \bbE \lt[ \one_{\sigma_r < 1} \bbE^{x_{\sigma_r}} \lt[ \e{I_0 ^{m-2}} \one_{K_{[2,m-3]}} \one_{\Tilde{x}_{m-1} \in K}  \rt] \rt] \\
            &\ge \frac{C_3}{\sqrt{2^d}} \bbP(\sigma_r < 1)    \bbE^{y} \lt[ \e{I_0 ^{m-1}} \one_{K_{[2,m-3]}} \one_{\Tilde{x}_{m-1} \in K}  \rt] \\
            &\ge \frac{C_3}{\sqrt{2^d}} \bbP(\sigma_R < 1)   \bbE^{y} \lt[ \e{I_0 ^{m-1}} \one_{K_{[2,m-3]}} \one_{\Tilde{x}_{m-1} \in K}  \rt].
            \qedhere
        \end{split}
    \end{equation}
\end{proof}

\begin{proof}[Proof of Proposition \ref{approx_markov}]
    Define $\sigma_L := \inf\{ t \ge i : x_t \in K\}$, $\sigma_R := \inf\{ t \ge j : x_t \in K\}$.
    By definition 
    $D_i = \{ i \le \sigma_L  \le i+1\}$ and the same for $D_j$.
    Applying the strong Markov property yields
    \begin{equation}
        \nn
        \begin{split}
            & \bbE \lt[ \e{I_0 ^T} \one_{\sigma_L \le i+1} \one_{K_{[i+1,j]}} \one_{\sigma_R \le j+1} \one_{x_T \in K} \rt] \\
            &\le C_1  \bbE \lt[ \e{I_0 ^{i-1}} \one_{\sigma_L \le i+1} 
            \e{I_{i+1} ^{j-1}}
            \one_{K_{[i+1,j]}} \one_{\sigma_R \le j+1}
            \e{I_{j+1} ^{T}}
            \one_{x_T \in K} \rt]\\
            &= C_1 \bbE \lt[ \e{I_0 ^{i-1}} \bbE^{x_{i-1}} \lt[ \one_{1 \le \sigma_L \le 2} \e{I_{2} ^{j-i}}
            \one_{K_{[2,j-i+1]}} \one_{\sigma_R \le j-i+2}
            \e{I_{j-i+2} ^{T-i+1}}
            \one_{\Tilde{x}_{T-i+1} \in K}
        \rt] \rt] \\
        &= C_1 \bbE \Big[ \e{I_0 ^{i-1}} \bbE^{x_{i-1}} \Big[ \one_{1 \le \sigma_L \le 2} \bbE^{\Tilde{x}_{\sigma_L}} \Big[  \e{I_{2-\sigma_L} ^{j-i-\sigma_L}} \times \\
        & \qquad \times 
            \one_{K_{[2-\sigma_L,j-i+1-\sigma_L]}} \one_{\sigma_R \le j-i+2-\sigma_L}
            \e{I_{j-i+2-\sigma_L} ^{T-i+1-\sigma_L}}
            \one_{\overline{x}_{T-i+1-\sigma_L} \in K} \Big] \Big] \Big] \\
            &= C_1 \bbE \Big[ \e{I_0 ^{i-1}} \bbE^{x_{i-1}} \Big[ \one_{1 \le \sigma_L \le 2} \bbE^{\Tilde{x}_{\sigma_L}} \Big[  \e{I_{2-\sigma_L} ^{j-i-\sigma_L}} \times \\
            & \qquad \times  
            \one_{K_{[2-\sigma_L,j-i+1-\sigma_L]}} \one_{\sigma_R \le j-i+2-\sigma_L} \bbE^{\overline{x}_{{j-i+1}}} \Big[
            \e{I_{1-\sigma_L} ^{T-j-\sigma_L}}
            \one_{x^\dag _{T-j-\sigma_L} \in K} \Big] \Big]\Big]\Big] \\
            &\le C_2 \bbE \Big[ \! \e{I_0 ^{i-1}} \!\! \bbE^{x_{i-1}} \Big[ \one_{1 \le \sigma_L \le 2} \bbE^{\Tilde{x}_{\sigma_L}} \Big[\!  \e{I_{2} ^{j-i}} \!\!
            \one_{K_{[2,j-i-1]}} \! \one_{\sigma_R \le j-i+2-\sigma_L} \bbE \Big[\!
            \e{I_{0} ^{T-j}} \!\!
            \one_{x^\dag _{T-j-\sigma_L} \in K} \Big] \Big]\Big]\Big] \!\! =:\!\! (*).\\
        \end{split}
    \end{equation}
    Here the last inequality follows from Proposition \ref{prop_zero_max}, re-introducing all terms in the exponential and noticing $\one_{K_{[2-\sigma_L,j-i+1-\sigma_L]}} \le \one_{K_{[2,j-i-1]}} $. Our goal is now to remove the stopping time $\sigma_L$ from the inner expectations. By doing this from right-to-left we can also replace $\one_{D_i}$ by $\one_{x_{i+1} \in K}$, which yields the desired result.
    
    In a first step we can estimate 
    $$ \bbE \lt[\e{I_{0} ^{T-j}}\one_{x^\dag _{T-j-\sigma_L} \in K} \rt]  \le C \bbE \lt[\e{I_{0} ^{T-j}}\one_{x^\dag _{T-j} \in K} \rt] $$
    using once again Proposition \ref{pin_move_prop}, which makes this term deterministic. Applying $3$ times Proposition \ref{move_pinning_to_left_prop} one can estimate
    $$ \bbE \lt[\e{I_{0} ^{T-j}}\one_{x^\dag _{T-j} \in K} \rt] \le C \bbE \lt[\e{I_{0} ^{T-j-3}}\one_{x^\dag _{T-j-3} \in K} \rt].$$
    This yields
    \begin{equation}
        \label{double_star}
        (*) \le C \bbE \Big[\!\e{I_{0} ^{T-j-3}}\one_{x^\dag _{T-j-3} \in K} \Big] \bbE \Big[\! \e{I_0 ^{i-1}} \bbE^{x_{i-1}} \Big[ \one_{1 \le \sigma_L \le 2} \bbE^{\Tilde{x}_{\sigma_L}} \Big[\!  \e{I_{2} ^{j-i}}
            \one_{K_{[2,j-i-1]}} \one_{\sigma_R \le j-i+2-\sigma_L} \Big]\Big]\Big].
    \end{equation}
    Moreover, for each $y \in \bbR^d$ and abbreviating $l = j-i +1$
    \begin{equation}
        \nn
        \begin{split}
            \bbE^{y} \lt[  \e{I_{2} ^{l-1}}
            \one_{K_{[2,l-2]}} \one_{\sigma_R \le l+1-\sigma_L}  \rt] &\le C \bbE^{y} \lt[  \e{I_{2} ^{l-2}}
            \one_{K_{[2,l-2]}} \bbE^{x_{l-2}} \lt[ \one_{2 - \sigma_L \le \sigma_R \le 3-\sigma_L}\rt] \rt] =: (**). 
        \end{split}
    \end{equation}
Continuing to estimate the second expectation, for all $z \in \bbR^d$ and $t \in [0,1]$
\begin{equation}
\label{existence_vs_endpoint_pin}
    \begin{split}
        \bbE^{z} \lt[ \one_{1 - t \le \sigma_R \le 2-t}\rt] &= \bbE^{z} \lt[ \one_{1 - t \le \sigma_R \le 2-t}  \frac{\bbP(x_{2-\sigma_R} \in K)}{\bbP(x_{2-\sigma_R} \in K)} \rt] \\
        &\le C \bbE^{z} \lt[ \one_{1 - t \le \sigma_R \le 2-t}  \bbP^{x_{\sigma_R}}(x_{2-\sigma_R } \in K) \rt] \\
        &= C \bbE^{z} \lt[ \one_{1 - t \le \sigma_R \le 2-t} \one_{x_2 \in K} \rt] \\
        &\le C \bbE^{z} \lt[\one_{x_2 \in K} \rt]
    \end{split}
\end{equation}
and so
$$(**) \le C \bbE^{y} \lt[  \e{I_{2} ^{l-2}}
            \one_{K_{[2,l-2]}} \one_{x_{l}} \in K \rt] \le C_2 \bbE^{y} \lt[  \e{I_{2} ^{l}}
            \one_{K_{[2,l-2]}} \one_{x_{l} \in K} \rt].$$
It is now possible to use Proposition \ref{technical_estimate_thm1_4} to find, for all $y \in K$
$$\bbE^{y} \lt[  \e{I_{2} ^{l}}
            \one_{K_{[2,l-2]}} \one_{x_{l} \in K} \rt] \le C_3 \bbE \lt[  \e{I_{0} ^{l+1}}
            \one_{K_{[3,l-2]}} \one_{x_{l+1} \in K} \rt].$$
Plugging this estimate into (\ref{double_star}) gives
$$(\ref{double_star}) \le C_2 \bbE \Big[\e{I_{0} ^{T-j-3}}\!\!\!\one_{x^\dag _{T-j-3} \in K} \Big] \bbE \lt[  \e{I_{0} ^{j-i+2}} \!\!\!
            \one_{K_{[3,j-i-1]}} \one_{x_{j-i+2} \in K} \rt] \bbE \lt[ \e{I_0 ^{i-1}} \bbE^{x_{i-1}} \lt[ \one_{1 \le \sigma_L \le 2}\rt]\rt]. $$
Finally, re-doing calculation (\ref{existence_vs_endpoint_pin}) yields the estimate
$$\bbE^{x_{i-1}} [ \one_{1 \le \sigma_L \le 2}] \le C \bbE^{x_{i-1}} [ \one_{x_2 \in K}].$$
Putting all this together leads to
\begin{equation}
    \nn
    \begin{split}
        &\bbE \lt[  \e{I_0 ^T} \one_{\sigma_L \le i+1} \one_{K_{[i+1,j-1]}} \one_{\sigma_R \le j+1} \one_{x_T \in K} \rt] \\
&\le C \bbE \lt[ \e{I_0 ^{i+1}} \one_{x_{i+1 \in K}}\rt] \bbE \lt[ \e{I_0 ^{j-i + 2}} \one_{K_{[3,j-i -1]}} \one_{x_{j-i + 2 \in K}}\rt] \bbE \lt[ \e{I_0 ^{T-j-3}} \one_{x_{T-j-3} \in K}\rt]. \qedhere
    \end{split}
\end{equation}
\end{proof}

\begin{appendix}
\section{Two basic results from quantum theory}
    \begin{theorem}
    \label{potential_well_parameter_existence}
    Fix $d \geq 3$. Let $V$ be a potential well with radius $r>0$ and take 
    $H_p = - \frac{\Delta}{2} + \delta V.$
    There exists $\delta^* >0$ such that for all $\delta \geq \delta^*$, $H_p$ has a ground state. The ground state energy $E_0$ is smaller than $-\delta 3/4$ for $\delta$ large enough.
\end{theorem}
\begin{proof}
    Choose any trial state that does not vanish inside the potential well.
\end{proof}

\begin{corollary}
    \label{ground_state_exponential_growth}
    Consider the Schrödinger operators from Theorem \ref{potential_well_parameter_existence} and fix $V$. There exists $c>0$ such that for $\delta \geq \delta^*$ and $T$ sufficiently large
    $$\bbE^\bbP \lt[ \exp\lt(-\delta \int_0 ^T \de s V(x_s) \rt) \rt] \geq \exp(cT).$$
    Moreover, we can choose $c = \delta/2$ for $\delta$ large enough.
\end{corollary}
\begin{proof}
    By Theorem \ref{potential_well_parameter_existence} it is possible to take $\delta > 0$ such that there exists a ground state with ground state energy $E_0 < 0$ for the Schrödinger operator
    $$H = - \frac{\Delta}{2} + \delta V,$$
    with $V(x) = - \one_{[0,r]}(\Vert x \Vert)$. We denote the corresponding normalized eigenvector by $\xi_0$ and choose it w.l.o.g. strictly positive. Take $\psi \in L^2(\bbR^d)$ such that 
    $$\la \psi, \e{- tH + tE_0} \psi \ra \rightarrow \la \psi, \xi_0 \ra ^2.$$
    By picking a $\psi \le 1$, quasi-concave and non-orthogonal to $\xi_0$ we apply Proposition \ref{prop_zero_max} and find
    \begin{equation}
        \nonumber
        \begin{split}
            \la \psi, \e{- tH + tE_0} \psi \ra &= \e{tE_0} \int \de x_0 \psi(x_0) \bbE^{x_0} \lt[ \e{\delta V_0 ^t (x)} \psi(x_t) \rt] \\
            &\leq \e{tE_0} c \bbE \lt[ \e{\delta V_0 ^t (x)} \psi(x_t) \rt] \\
            &\leq \e{tE_0} c \bbE \lt[ \e{\delta V_0 ^t (x)}\rt],
        \end{split}
    \end{equation}
    which shows the claim since the L.H.S. in the above display converges for $t \rightarrow \infty$ to something greater $0$.
\end{proof}

\section{Criteria for the existence of ground states}
Let us recall the following fact from measure theory:
\begin{proposition}
\label{atom}
    Let $\mu$ be a finite regular Borel measure on $\bbR$, with $\inf\text{supp} \: \mu = 0$. Then,
    \begin{equation}
         \label{inf supp atom criterion}
         \liminf_{t \rightarrow \infty} \frac{\int \e{-tx} \mu(\de x)}{\lt( \int \e{-2tx} \mu(\de x) \rt)^{1/2}} > 0 \: \:\:\:\:\: \iff \:\:\:\:\:\: \mu (\{ 0 \}) > 0.
    \end{equation}
\end{proposition}
\begin{proof}
    The if part is clear by dominated convergence. For the only if part, assume that $\mu(\{0\}) = 0$. 
Since $\mu$ is regular, $0 = \lim_{n \to \infty} \mu((-1/n,1/n))$, and 
\[
\int \e{-tx} \mu(\dd x) \leq \int_{-1/n}^{1/n} \e{-tx} \mu(\dd x) + \e{-t/n} \mu(\bbR).
\]
The first term above is bounded by 
$\left(\int_{-1/n}^{1/n} \e{-2tx} \mu(\dd x) \right)^{1/2} \mu((-1/n,1/n))^{1/2}$ due to the Cauchy-Schwarz 
inequality. Thus the quotient of this term with the denominator in \eqref{inf supp atom criterion} can be made 
arbitrarily small by taking $n$ large. The second term, divided by the denominator of \eqref{inf supp atom criterion}, 
converges to zero for each $n$, as $t \to \infty$. 
To see this, it suffices to note that $\int \e{2 t x} \mu(\dd x) \geq \e{-t/2n} \mu([0,1/4n])$, and the latter expression 
is nonzero since $\inf\supp \mu = 0$. Thus the `only if' part of the claim holds.
\end{proof}
Define 
$$\gamma_{\psi}(t) := \frac{(\psi,\e{-tH}\psi)}{(\psi,\e{-2tH}\psi)^{1/2}},$$
where $H$ is a self-adjoint operator on the Hilbert space $\caH = L^2 (X, \de m)$ in which $X$ is a measure space and $m$ is regular. We denote the inner product of $\caH$ by $(\cdot,\cdot)$.
For the next statement, recall the notion of a positivity improving operator: let $A$ be a bounded operator on  
$L^2(X, \dd \mu)$ for some measure space $(X, \mu)$. If 
$f \geq 0$ $\mu$-almost everywhere implies $Af > 0$ 
$\mu$-almost everywhere, then $A$ is called 
{\em positivity improving}. It is known that the particle-field Hamiltonian \eqref{hamiltonian} is positivity improving, see e.g.\ \cite{HiLo20}, chapter 2. Therefore the following result applies to our situation.   

\begin{theorem}
\label{ground_state_options}
Let $H$ be an operator in $\caH$, bounded below, and assume that $\e{-tH}$ is positivity improving for all large enough $t>0$. Then, the following are equivalent:
\begin{enumerate}
    \item $H$ has a unique, strictly positive ground state.
    \item There exists $\psi \in \caH$ with $\psi \geq 0$ and $\liminf\limits_{t \rightarrow \infty} \gamma_\psi (t) > 0$.
    \item For all $\psi \in \caH$ with $\psi \geq 0$ and $m(\psi > 0) >0$, $\liminf\limits_{t \rightarrow \infty} \gamma_{\psi} (t) > 0$ holds.
\end{enumerate}
\end{theorem}
\begin{proof}
    $(1) \Rightarrow (3):$ By the spectral theorem, $\lim_{t \to \infty} (\psi, \e{-t H} \psi) = (\psi,\psi_0)$,
where $\psi_0$ is the ground state for $H$. As $\psi_0$ is strictly positive by the Perron-Frobenius theorem, 
(3) holds.\\
$(3) \Rightarrow (2)$ is trivial.\\
$(2) \Rightarrow (1)$: For $\psi$ as in (2), let $\mu_\psi$ be the spectral measure associated with $\psi$. 
Then $(\psi, \e{-t H} \psi) = \int \e{-t x} \mu_\psi(\dd x)$. $\mu_\psi$ is a regular Borel measure with 
$\mu_\psi(\bbR) = \| \psi \|^2$, and $E_\psi := \inf \rm{supp} \: \mu_\psi \ge \inf\rm{spec} \: H = 0$. 
Now (2) and Proposition \ref{atom}
imply that $0 < \mu_{\psi} ( \{E_\psi\}) = \| P_{E_\psi} \psi \|^2$, where $P_{E_\psi}$ is the projection
onto the spectral subspace of $H$ corresponding to $\{E_\psi\}$. Consequently, $E_\psi$ is an eigenvalue,
and $P_{E_\psi} \psi$ is an eigenfunction. It remains to show that $E_\psi = 0$. For showing this, first note 
that we may  assume $\psi > 0$. Namely, recall that for any $\phi \in \caH$, 
\[
\inf \rm{supp} \: \mu_\phi = - \lim_{t \to \infty} \frac{1}{t} \ln \int \e{-tx} \mu_\phi (\dd x),
\]
where $\mu_\phi$ is again the spectral measure associated to $\phi$. 
Thus clearly $\inf \rm{supp} \: \mu_\psi = \inf \rm{supp} \: \mu_{\e{- \delta H} \psi}$ for each $\delta > 0$, and 
since $\e{-tH}$ is positivity improving, we have found a strictly positive function such that the 
infimum of its support is equal to that of $\mu_\psi$. Now we show that for any strictly positive function $\psi$,
$E_\psi = \inf \rm{supp} \: \mu_\psi = \inf \rm{spec} \: H = 0$. By spectral theory, 
$E_0 = \inf \{ \inf\rm{supp} \: \mu_\phi: \phi \in A \}$, where $A$ is a dense subset of $\caH$. We choose
$A = \bigcup_{M > 0} \bigcup_{n > 0} A_{M,n}$, with 
\[
A_{M,n} = \{ \phi \in L^2: \sup_x |\phi(x)| \leq M, \phi(x) = 0 \text{ if } \psi(x) < 1/n \}
\]
(note that $\psi > 0$ and the regularity of $m$ imply that $A$ is dense). Now for any $\eps > 0$ choose 
$\phi \in A_{M,n}$ 
with $\inf\supp \mu_\phi < \eps$. Since $\e{-tH}$ improves positivity, we have 
\[
\int \e{-tx} \mu_{|\phi|}(\dd x) \geq \int \e{-tx} \mu_{\phi}(\dd x),
\]
and thus $\inf\rm{supp} \: \mu_{|\phi|} \le	\inf\rm{supp} \: \mu_\phi < \eps$. Now 
$0<|\phi(x)| \leq n M \psi(x)$, and thus 
\[
(|\phi|, \e{-tH}|\phi|) \leq n^2 M^2 (\psi,\e{-tH} \psi).
\] 
Thus $\inf\rm{supp} \: \mu_{\psi} < \eps$, and a ground state exists. It is unique and can be chosen strictly positive by the Perron-Frobenius theorem.
\end{proof}

\end{appendix}

\bigskip
{\bf Acknowledgments:} This research was supported by DFG grant No 535662048. V.B. wants to thank Erwin Bolthausen for many motivating discussions about the topic over the years, and Fumio Hiroshima for a helpful discussion about enhanced binding in various models. We also thank Benjamin Hinrichs for useful comments.




\end{document}